\documentclass[11pt,twocolumn,notitlepage,tightenlines,nofootinbib,superscriptaddress]{revtex4-2}
\usepackage{amsmath,amsxtra,amssymb,amsthm,amsfonts,bm}

\usepackage{comment}
\usepackage{tikz}
\usepackage{amsmath}
\usepackage{amssymb}  
\usepackage{amsthm}
\usepackage{enumitem}
\usepackage{amsfonts}
\usepackage{graphicx}
\usepackage{bbm}
\usepackage{url}
\usepackage{stmaryrd}
\usepackage{hyperref}
\usepackage{braket}
\usepackage{ upgreek }
\usepackage{dsfont}
\usepackage{makecell}
\usepackage{natbib}
\usepackage{fancyref} %
\usepackage{algorithm}
\usepackage{algpseudocode}
\usepackage[utf8]{inputenc}

\usepackage{float}
\newfloat{algorithm}{t}{lop}

\usepackage{hyperref}

\theoremstyle{plain} 

\newtheorem{theorem}{Theorem}[section]
\newtheorem{thm}{Theorem}[section]

\newtheorem{lem}{Lemma}[section]

\newtheorem{cor}{Corollary}[section]

\newtheorem{prop}{Proposition}[section]

\newtheorem{defi}{Definition}[section]

\newtheorem{exmp}{Example}[section]
\theoremstyle{remark}
\newtheorem{rem}{Remark}[section]

\newcommand*{\fancyrefthmlabelprefix}{thm}
\newcommand*{\fancyreflemlabelprefix}{lem}
\newcommand*{\fancyrefcorlabelprefix}{cor}
\newcommand*{\fancyrefdefilabelprefix}{defi}
\frefformat{plain}{\fancyreflemlabelprefix}{lemma\fancyrefdefaultspacing#1}
\Frefformat{plain}{\fancyreflemlabelprefix}{Lemma\fancyrefdefaultspacing#1}
\frefformat{plain}{\fancyrefthmlabelprefix}{theorem\fancyrefdefaultspacing#1}
\Frefformat{plain}{\fancyrefthmlabelprefix}{Theorem\fancyrefdefaultspacing#1}
\frefformat{plain}{\fancyrefcorlabelprefix}{corollary\fancyrefdefaultspacing#1}
\Frefformat{plain}{\fancyrefcorlabelprefix}{Corollary\fancyrefdefaultspacing#1}
\frefformat{plain}{\fancyrefdefilabelprefix}{definition\fancyrefdefaultspacing#1}
\Frefformat{plain}{\fancyrefdefilabelprefix}{Definition\fancyrefdefaultspacing#1}
\newcommand*{\fancyrefalglabelprefix}{alg}
\newcommand*{\frefalgname}{algorithm}
\newcommand*{\Frefalgname}{Algorithm}
\frefformat{plain}{\fancyrefalglabelprefix}{%
  \frefalgname\fancyrefdefaultspacing#1%
}%
\Frefformat{plain}{\fancyrefalglabelprefix}{%
  \Frefalgname\fancyrefdefaultspacing#1%
}%

\newcommand*{\fancyrefapplabelprefix}{app}
\newcommand*{\frefappname}{appendix}
\newcommand*{\Frefappname}{Appendix}
\frefformat{plain}{\fancyrefapplabelprefix}{%
  \frefappname\fancyrefdefaultspacing#1%
}%
\Frefformat{plain}{\fancyrefapplabelprefix}{%
  \Frefappname\fancyrefdefaultspacing#1%
}%

\definecolor{Green}{HTML}{00AD69}  %

\def\beq{\begin{equation}}
\def\eeq{\end{equation}}
\def\bq{\begin{quote}}
\def\eq{\end{quote}}
\def\ben{\begin{enumerate}}
\def\een{\end{enumerate}}

\def\bit{\begin{itemize}}
\def\eit{\end{itemize}}

\def\out{{\operatorname{out}}}

\def\l|{\left|}
\def\r|{\right|}

\newcommand\R{\mathbbm{R}}
\newcommand\N{\mathbbm{N}}

\newcommand\cB{\mathcal{B}}

\newcommand{\cD}{\mathcal{D}}

\newcommand{\cL}{\mathcal{L}}

\newcommand{\ketbra}[1]{|#1\rangle\langle#1|}
\newcommand{\tr}{{\operatorname{tr}}}

\newcommand{\daniel}[1]{\textcolor{orange}{[daniel] #1}}

\usepackage{geometry}
\newcommand{\cS}{\mathcal{S}}

\newcommand{\cO}{\mathcal{O}}

\newcommand{\cH}{\mathcal{H}}

\newcommand{\eps}{\epsilon}
\newcommand{\cN}{\mathcal{N}}

\begin{document}

\author{Giacomo De Palma}
\affiliation{Department of Mathematics, University of Bologna, 40126 Bologna, Italy}
\email[Giacomo De Palma ]{giacomo.depalma@unibo.it}
\author{Milad Marvian}
\affiliation{Department of Electrical \& Computer Engineering and Center for Quantum Information and Control, University of New Mexico, Albuquerque, NM 87131, USA}
\email[Milad Marvian ]{mmarvian@unm.edu}

\author{Cambyse Rouz\'{e}}
 \affiliation{Zentrum Mathematik, Technische Universit\"{a}t M\"{u}nchen, 85748 Garching, Germany}
\email[Cambyse Rouz\'{e} ]{cambyse.rouze@tum.de}

\author{Daniel Stilck Fran\c{c}a}
\affiliation{QMATH, Department of Mathematical Sciences, University of Copenhagen, Universitetsparken 5, 2100 Copenhagen, Denmark}
\affiliation{Univ Lyon, ENS Lyon, UCBL, CNRS, Inria, LIP, F-69342, Lyon Cedex 07, France}
\email[Daniel Stilck Fran\c ca ]{daniel.stilck\_franca@ens-lyon.fr}

\title[Limitations of VQAs: a quantum optimal transport approach]{Limitations of variational quantum algorithms: \\a quantum optimal transport approach
}

\begin{abstract}
The impressive progress in quantum hardware of the last years has raised the interest of the quantum computing community in harvesting the computational power of such devices. However, in the absence of error correction, these devices can only reliably implement very shallow circuits or comparatively deeper circuits at the expense of a nontrivial density of errors.
In this work, we obtain extremely tight limitation bounds for standard NISQ proposals in both the noisy and noiseless regimes, with or without error-mitigation tools. The bounds limit the performance of both circuit model algorithms, such as QAOA, and also continuous-time algorithms, such as quantum annealing. In the noisy regime with local depolarizing noise $p$, we prove that at depths $L=\cO(p^{-1})$ it is exponentially unlikely that the outcome of a noisy quantum circuit outperforms efficient classical algorithms for combinatorial optimization problems like Max-Cut. 
Although previous results already showed that classical algorithms outperform noisy quantum circuits at constant depth, these results only held for the expectation value of the output. Our results are based on newly developed quantum entropic and concentration inequalities, which constitute a homogeneous toolkit of theoretical methods from the quantum theory of optimal mass transport whose potential usefulness goes beyond the study of variational quantum algorithms.

\end{abstract}
\maketitle

\section{Introduction}

The last years have seen remarkable progress in both the size and quality of available quantum devices, reaching the point in which even the best classical computers cannot easily simulate them~\cite{Arute2019,Zhong2020,Scholl2021,Ebadi2021}. 
In spite of these achievements, current devices lack error correction and, thus, are inherently noisy. 
Considering the significant overheads required to implement error correction~\cite{campbell2017roads, Roffe_2019}, this has raised the quantum computing community's interest in investigating whether such noisy quantum devices can nevertheless outperform classical computers at tasks of practical interest~\cite{Preskill2018}.

One class of algorithms that are considered suited for this task is variational quantum algorithms~\cite{Bharti_2022,Cerezo_2021}. 
In most cases, these hybrid quantum-classical algorithms work by optimizing the parameters of a shallow quantum circuit to minimize a cost function~\cite{Bharti_2022,Cerezo_2021}. 
Prominent examples of such algorithms include the variational quantum eigensolver (VQE)~\cite{Peruzzo_2014} and the quantum approximate optimization algorithm (QAOA)~\cite{Farhi2014}. As variational algorithms only require the implementation of shallow circuits and simple measurements, it was expected that they could unlock the computational potential of near-term devices.

However, recent results have highlighted several obstacles to achieving a practical quantum advantage through variational quantum algorithms. For instance, some works have shown that optimizing the parameters of the circuit is computationally expensive in various settings~\cite{training_hard,mcclean2018barren,wang2021noiseinduced,anschuetz2022critical}. Other works have shown that constant depth quantum circuits cannot outperform classical algorithms for certain combinatorial optimization problems~\cite{bravyi_obstacles_2020,farhi_quantum_2020,moosavian_limits_2021,chou2021limitations}. 
Furthermore, it has been observed~\cite{wang2021noiseinduced,Franca2021a,gonzalez2022error} that such variational quantum algorithms are less robust to noise than previously expected: already a small density of errors is sufficient to ensure that classical algorithms outperform the noisy device.

In this article, we further investigate the limitations of variational quantum algorithms. Our contributions are two-fold: First, we obtain extremely tight limitation bounds for standard NISQ proposals in both the noisy and noiseless regimes, with or without error mitigation tools. Second, we provide a new homogeneous toolkit of theoretical methods whose potential usefulness goes beyond the present topic of variational quantum algorithms. Our methods originate from the emerging field of quantum optimal transport~\cite{Carlen_2014,Carlen2017,Carlen_2019,Rouze2019,datta2020relating,Gao2020,DePalma2021,de2020quantum,depalma2021quantum,gao2021ricci}. As we will see, optimal transport techniques have the combined advantages of simultaneously simplifying, unifying, and qualitatively refining previously known statements regarding fundamental properties of the output state of shallow and noisy circuits.

\subsection*{Limitations of noisy variational quantum algorithms}
More precisely, we obtain two new complementary sets of results providing a better understanding of the limitations of variational quantum algorithms both at very shallow depths, when the effect of noise is negligible, and for a small density of errors. In \autoref{variancebound}, we first derive new properties for the output probability of (potentially noisy) shallow quantum circuits initiated in the state $|0\rangle^{\otimes n}$ and after measurement in the computational basis. These findings directly improve upon celebrated recent results on the limitation of certain variational quantum algorithms to solve the Max-Cut problem for certain classes of bipartite $D$-regular graphs. We prove that QAOA requires at least logarithmic in system size depth $L$ to outperform efficient classical algorithms in some instances \cite{bravyi_obstacles_2020,farhi_quantum_2020,moosavian_limits_2021,eldar2017local}:
\begin{equation}\label{eq:boundLus}
L\ge \frac{1}{2\log\left(D+1\right)}\log\frac{n}{576}\,.
\end{equation}
We notice that our bound in \eqref{eq:boundLus} exponentially improves upon the dependence on the degree $D$ of the graph previously found in \cite{bravyi_obstacles_2020}. For instance, for $D=55$ (the minimum value for which Ref. \cite{bravyi_obstacles_2020} can prove that shallow quantum circuits cannot outperform the classical algorithm by Goemans and Williamson), our bound implies that the QAOA requires a depth larger than $1$ as soon as $n\gtrapprox10^{6}$, whereas Ref. \cite{bravyi_obstacles_2020} gave $n\gtrapprox10^{54}$.

Next, \autoref{sec:lci} is concerned with the concentration profile of the output measure of noisy circuits at any depth $L=\Omega(1)$ for simple noise models, e.g.~layers of circuits interspersed by layers of one-qubit depolarizing noise of parameter $p$. For instance, we are able to prove with realistic depolarizing probability $p=0.1$ applied independently to each qubit, the number of vertices for the graph has to be smaller than $10^9$ in order for the noisy algorithm to outperform the best known classical algorithm (see \autoref{thm:noisyMax-Cut}). 
Moreover, we prove that at depths $L=\cO(p^{-1})$ it is exponentially unlikely that the outcome of a noisy quantum circuit outperforms efficient classical algorithms for combinatorial optimization problems like Max-Cut. 
Although previous results already showed that noisy quantum circuits are outperformed by classical algorithms at constant depth~\cite{Franca2021a}, their results only held for the expectation value of the output.
In contrast, our methods imply that the probability of observing a single string with better energy than the one outputted by an efficient classical algorithm is exponentially small in the number of qubits. 
This is a significantly stronger statement, although at the cost of slightly worse constants~\cite{Franca2021a}.

In addition, in \autoref{mitigationsec}, we show that certain error mitigation protocols cannot reverse our conclusions unless we allow for an exponential number of samples in the number of qubits. First, in~\autoref{sec:virtual_cooling} we show that virtual distillation or cooling protocols~\cite{Huggins2021,Koczor2021} only have an exponentially small success probability at constant depth. Furthermore, for mitigation procedures that have as their goal to estimate expectation values of observables, we show stringent limitations at $\cO(\log(n))$ depth in~\autoref{sec:weakerror}. At this depth, any error mitigation procedure that takes as input $m=\textrm{poly}(n)$ copies of the output of a noisy quantum circuit is exponentially unlikely to yield an estimate that deviates significantly from the estimate we would obtain by providing $m$ copies of a trivial product state as input. Thus, the copies of the noisy quantum circuit do not provide significantly more insights than sampling from trivial product states. Our results strengthen recent results on limitations of error mitigation~\cite{Takagi2021,Wang2021} both in terms of the required depth for them to apply and by providing concentration inequalities instead of results in expectation.

\subsection*{Quantum optimal transport toolkit} The second main contribution of the present article is the development of a new set of simple methods from quantum optimal transport whose potential use is likely to exceed the problem of finding tighter limitations on variational quantum algorithms. Our first main tool leading to the results of \autoref{variancebound} is an optimal transport inequality introduced by Milman \cite{Milman01} in his study of the concentration and isoperimetric profile of probability measures on Riemannian manifolds with positive curvature (see also \cite{Jourdain2012,ledoux18,Liu2020} for discussions on some related optimal transport inequalities). Adapted to the present setting of $n$-bit strings $\{0,1\}^n$ endowed with the Hamming distance $d_H(x,y):=\sum_{i=1}^n|x_i-y_i|$, Milman's so called \textit{$(2,\infty)$-Poincar\'e inequality} is a property of a probability measure $\mu$ on the set $\{0,1\}^n$ which asks for the existence of a constant $C>0$ such that, for any function $f:\{0,1\}^n\to\mathbb{R}$,
\begin{align}\label{2inftypoincareintro}
    \operatorname{Var}_\mu(f)\,\le C\,n\,\| f\|_L^2, 
\end{align}
where $\|f\|_L:=\sup_{x\ne y}\frac{|f(x)-f(y)|}{d_H(x,y)}$ denotes the Lipschitz constant of $f$ with respect to the Hamming distance. Besides its natural application to bounding the probability that the function $f$ deviates from its mean by means of Chebyshev's inequality, namely
\begin{align}\label{concentration}
    \mathbb{P}_\mu(|f-\mathbb{E}_\mu[f]|\ge \sqrt{n}\,r)\,\le \frac{C\|f\|_L^2}{r^2}\,,
\end{align}
the $(2,\infty)$-Poincar\'{e} inequality further implies by duality the following symmetric concentration inequality: for any two sets $S_1,S_2\subset \{0,1\}^n$ such that $\mu(S_1),\mu(S_2)\ge\mu_0>0$,
\begin{align}\label{isop_intro}
    d_H(S_1,S_2)\le\,3\,\sqrt{ \frac{Cn}{\mu_0}}\,.
\end{align}
For instance, we prove in \autoref{noiselesslight-cone} that in the case of a noiseless circuit, the output measure $\mu_{\operatorname{out}}$ satisfies the $(2,\infty)$-Poincar\'{e} inequality with constant $C=B^2$, where $B$ denotes the light-cone of the circuit, i.e.~the maximal amount of output qubits being influenced by the value of an arbitrary input qubit through the application of the circuit. In that case, the resulting symmetric concentration inequality \eqref{isop_intro} quadratically improves over the one previously derived in \cite[Corollary 43]{eldar2017local}:
\begin{align*}
d_H(S_1,S_2)\le 4\,\frac{\sqrt{n}\,B^{1.5}}{\mu_0}\,.    
\end{align*}
Moreover, the $(2,\infty)$-Poincar\'{e} inequality turns out to be a very simple and versatile tool as compared to the non-trivial proof of  \cite[Corollary 43]{eldar2017local} which required the use of Chebyshev polynomials and approximate projections. Moreover, it can be very easily adapted to noisy shallow quantum circuits and continuous-time local Hamiltonian evolutions. In this latter setting, it unifies and refines the main results of \cite{moosavian_limits_2021}. 

The tools described in the previous paragraph are adapted to the study of quantum circuits of depth $L=\mathcal{O}(\log(n))$ and related short-time continuous-time evolutions. In contrast, our second set of fundamental results in \autoref{sec:lci} concerns the concentration profile of the output measure of noisy circuits at any depth $L=\Omega(1)$ for simple noise models, e.g.~layers of circuits interspersed by layers of one-qubit depolarizing noise of parameter $p$. In this case, we appeal to recently developed tools such as contraction coefficients for Sandwiched R\'{e}nyi divergences~\cite{MuellerLennert2013,Wilde2014} in order to prove that the probability under the output measure $\mu_{\operatorname{out}}$ of the circuit that an arbitrary $n$-bit function $f:\{0,1\}^n\to \mathbb{R}$ deviates from its mean by a constant fraction $an$ of the total number of qubits satisfies the sub-Gaussian property: 
\begin{align}
\mathbb{P}_{\mu_{\operatorname{out}}}\big(|f-\langle f\rangle_{\mu_{\operatorname{out}}}|\ge an\big)\le K\,e^{-\frac{ca^2n}{\|f\|_L^2}}   
\end{align}
for some constants $K,c>0$ and $a\ge a_0\ge 0$. Interestingly, such strong concentration inequalities are known to be equivalent to a strengthening of the $(2,\infty)$-Poincar\'{e} inequality \eqref{2inftypoincareintro} known as \textit{transportation-cost} inequality \cite{marton1986simple,bobkov1999exponential}. The latter states that for any measure $\nu$ that is absolutely continuous with respect to $\mu$, 
\begin{align}\label{equ:tc_defi}
    W_1(\nu,\mu)^2\le C'\,n\,D(\nu\|\mu)\,,
\end{align}
where $D(\nu\|\mu)$ denotes the relative entropy of $\nu$ with respect to $\mu$, whereas 
\begin{equation}
W_1(\nu,\mu):=\sup_{\|f\|_L\le 1}\left(\mathbb{E}_\nu[f]-\mathbb{E}_\mu[f]\right)
\end{equation}
is the Wasserstein distance of order $1$ between $\nu$ and $\mu$, also called Monge--Kantorovich distance or earth mover's distance.

In summary, our results clearly illustrate the potential of optimal transport methods such as the $(2,\infty)$-Poincar\'{e} and the stronger transportation-cost inequality to study the performance of variational algorithms. %
We also believe that the discussed methods can have broad applications beyond that of understanding the computational power and limitations of near-term quantum devices.

Indeed, variations of inequalities like~\eqref{equ:tc_defi}  and~\eqref{2inftypoincareintro} have recently found applications in different areas of quantum information theory. For instance, in~\cite{rouze2021learning} they were used to obtain exponential improvements for the sample complexity in quantum tomography. In~\cite{depalma2021quantum} they were used to derive concentration bounds for commuting Gibbs states and show a strong version of the eigenstate thermalization hypothesis, a topic of intense research in physics. Thus, we believe that the new inequalities and techniques developed here could pave the way to extending such results to larger classes of states.

\section{Notations and definitions}\label{notationsdefs}

In this section, we introduce the main concepts discussed in the remaining of the paper. We also refer the reader to \autoref{notationsappendix} for a complete list of notations.

\subsection{Basic notions}
Given a set $V$ of $|V|=n$ qudits, we denote by $\cH_V=\bigotimes_{v\in V}\mathbb{C}^d$ the Hilbert space of $n$-qudits and by $\cB(\cH_V)$ the algebra of linear operators on $\cH_V$. 
$\mathcal{O}_V$ corresponds to the self-adjoint linear operators on $\cH_V$, whereas $\mathcal{O}^T_V\subset \mathcal{O}_V$ is the subspace of traceless self-adjoint linear operators. $\mathcal{O}_V^+$ denotes the subset of positive semidefinite linear operators on $\cH_V$ and $\mathcal{S}_V\subset \mathcal{O}_V^+$ denotes the set of quantum states. 
Similarly, we denote by $\mathcal{P}_V$ the set of probability measures on $[d]^V$.
For any subset $A\subseteq V$, we use the standard notations $\mathcal{O}_A, \mathcal{S}_A\ldots$ for the corresponding objects defined on subsystem $A$. 
Given a state $\rho\in\mathcal{S}_V$, we denote by $\rho_A$ its marginal on subsystem $A$. 
For any region $A\subset V$, the identity on $\mathcal{O}_{A}$ is denoted by $\mathbb{I}_A$, or more simply $\mathbb{I}$. 
Given an observable $O$, we define $\langle O\rangle_\sigma=\tr\left[ \sigma O\right]$.
We denote the probability of measuring an eigenvalue of $O$ greater than $a\in \mathbb{R}$ in the state $\sigma$ as $\mathbb{P}_\sigma(O\ge a)$.
Given two probability measures $\mu,\nu$ over a common measurable space, $\mu<\!<\nu$ means that $\mu$ is absolutely continuous with respect to $\nu$ and $\frac{d \mu}{d \nu}$ denotes the corresponding Radon-Nikodym derivative.

\subsection{Wasserstein distance}
We will make extensive use of notions of quantum optimal transport. The Lipschitz constant of the self-adjoint linear operator $H\in\mathcal{O}_V$ is defined as \cite[Section V]{de2020quantum}: 
\begin{align}\label{lipschitzdef}
    \|H\|_L:=2\max_{v\in V}\min_{H_{v^c}}\left\|H- H_{v^c}\otimes \mathbb{I}_v\right\|_\infty
\end{align}
where the infimum above is taken over operators $H_{v^c}\in\mathcal{O}_{V\setminus\{v\}}$ which do not act on $v$. Lipschitz observables, that is those $H\in \mathcal{O}_V$ such that $\|H\|_L=\mathcal{O}(1)$, capture extensive properties of a quantum system. They include (i) few-body and/or geometrically local observables; (ii) quasi-local observables; and even (iii) observables of the form $O=\sum_{i=1}^{\widetilde{n}}O_i$, where $\|O_i\|_\infty\le 1$ and $\operatorname{supp}(O_i)\cap \operatorname{supp}(O_j)=\emptyset$ for $i\ne j$. It is worth mentioning that the latter are considered in the fundamental problem in quantum statistical mechanics regarding the equivalence between the micro-canonical and canonical ensembles \cite{brandao2015equivalence,Tasaki2018,PhysRevLett.124.200604,depalma2021quantum}.

The quantum $W_1$ distance proposed in Ref. \cite{de2020quantum} admits a dual formulation in terms of the above quantum generalization of the Lipschitz constant: the quantum $W_1$ distance between the states $\rho,\omega\in\mathcal{S}_V$ is expressed as \cite[Section V]{de2020quantum}:
\begin{align*}
  W_1(\rho,\omega)  =\max\left\{\tr\left[\left(\rho-\omega\right)H\right]:\|H\|_L\le 1\right\}\,.
\end{align*}
Whereas the trace distance measures the global distinguishability of states, the Wasserstein distance measures distinguishability w.r.t. extensive, quasi-local observables.

\subsection{Local quantum channels}
In this work, we consider evolutions provided with a local description.
 \begin{defi}\label{defi:circuit}
A (noisy) quantum circuit $\cN_V$ on $n$ qudits of depth $L$ is a product of $L$ layers $\cN_1$, \ldots $\cN_L$, where each layer $\cN_{\ell}$ can be written as a tensor-product of  quantum channels $\cN_{\ell,e}$ acting on a set $e\subset V$ of vertices:
 \begin{align}\label{equ:localchannel}
     \cN_V=\prod_{\ell\in [L]}\,\bigotimes_{e\in E_\ell}\,\mathcal{N}_{\ell,e}\,\equiv \prod_{\ell\in [L]}\cN_{\ell} \,,
 \end{align}
 for some sets $E_\ell$ of disjoint subsets of vertices.
 The circuit is called unitary (or noiseless) whenever each of the channels $\mathcal{N}_{\ell,e}$ is unitary.
  We call the set $\{E_l\}_{l\in[L]}$ the architecture of the circuit $\cN_V$ and denote $E=\cup_{\ell}E_\ell$.
 \end{defi}
A key concept associated to the notion of a local evolution is that of a light-cone. In the case of a quantum circuit, the light-cone of a vertex $v\in V$ is the smallest set of vertices $I^{\cN_V}_v\equiv I_v\subseteq  V$ such that for any observable $O\in\mathcal{O}_V$ such that $\mathrm{Tr}_vO=0$ we have $\operatorname{tr}_{I_v}(\cN_V(O))=0$. 
We then denote the light-cone of the circuit by $I_{\cN_V}:=\max_{v\in V}|I_v|$. In \autoref{sec:continuous}, we extend this notion to the case of a continuous-time Hamiltonian evolution, where light-cones are defined thanks to Lieb-Robinson bounds \cite{lieb1972finite}. 

\section{Concentration at the output of short-time evolutions}\label{variancebound}

In this section, we obtain concentration inequalities for the outputs of short-time evolutions. Our main tool is an inequality between the variance and the Lipschitz constant of an observable. For any $O\in \cB(\cH_V)$ and $\omega\in\mathcal{S}_V$, the variance of $O$ in the state $\omega$ is defined as 
\begin{align*}
    \operatorname{Var}_\omega(O):=\tr\Big[\omega\,|O-\tr[\omega O]\,\mathbb{I}|^2\Big]=\langle\, |\,O-\langle O\rangle_\omega\,\mathbb{I}\,|^2\rangle_\omega \,.
\end{align*}
We denote the KMS inner product associated to the state $\sigma$ as $\langle A,B\rangle_\sigma:=\tr\big[A^\dagger\sigma^{\frac{1}{2}}B\sigma^{\frac{1}{2}}\big] $, and its corresponding norm as $\|H\|_{\sigma}$. We have for any $H\in\mathcal{O}_V$ that \cite[Equation (20)]{Temme2010}:
\begin{align}\label{KMStoGNS}
\|H-\tr[\sigma H]\,\mathbb{I}\|_{\sigma}^2\le \operatorname{Var}_\sigma(H)\,.
\end{align}
With a slight abuse of notation, we will use the same terminology for the analogous functionals for classical probability distributions. 

In analogy with the classical literature \cite{Milman01}, we say that a state $\sigma$ satisfies a $(2,\infty)$-Poincar\'{e} inequality of constant $C>0$ if for any $O\in\mathcal{O}_V$,
\begin{align}\label{eq:defPoinc}
    \operatorname{Var}_\sigma(O)\le C\,|V|\,\|O\|_L^2\,.
\end{align}
    For instance, tensor product states $\rho\equiv\bigotimes_{v\in V}\rho_v$ satisfy the $(2,\infty)$-Poincar\'{e} inequality with constant $C=1$ (see \autoref{appendixLipschitzcontrol}). The main motivation to introduce these inequalities are the following direct consequences of the $(2,\infty)$-Poincar\'{e} inequality. We leave their proof to \autoref{app:proofThm3.1}.

\begin{theorem}\label{thm:poicarconsequ}
Assume that the state $\sigma$ satisfies a $(2,\infty)$-Poincar\'{e} inequality of constant $C>0$. Then,

\begin{enumerate}[wide, labelwidth=!, labelindent=0pt]
\item \underline{Non-commutative transport-variance inequality}: for any two states $\rho_1,\rho_2\in\mathcal{S}_V$ with corresponding densities $X_j:=\sigma^{-\frac{1}{2}}\rho_j\sigma^{-\frac{1}{2}}$,
\begin{align*}
&W_1(\rho_1,\rho_2)\le \,\sqrt{C|V|\,} \big(\|X_1-\mathbb{I}\|_{\sigma}+\|X_2-\mathbb{I}\|_{\sigma}\big)\,.
\end{align*}

\item\label{thm:poicarconsequii} \underline{Measured transport-variance inequality}: denote by $\mu_\sigma\in \mathcal{P}_V$ the probability measure induced by the measurement of $\sigma$ in the computational basis. Then, for any probability measure $\nu<\!<\mu_\sigma$, 
\begin{align*}
W_1(\nu,\mu_\sigma)\le \sqrt{C\,|V|\,\operatorname{Var}_{\mu_\sigma}(d\nu/d\mu_\sigma)}\,.
\end{align*}
Moreover, for any two sets $A,B\subset [d]^V$, their Hamming distance $d_H(A,B)$ satisfies the following symmetric concentration inequality:
\begin{align}\label{classisoperimetric}
    d_H(A,B)\le \sqrt{C\,|V|}\,\big(\mu_\sigma(A)^{-\frac{1}{2}}+\mu_\sigma(B)^{-\frac{1}{2}}\big)\,.
\end{align}

\item \underline{Concentration of observables}: for any observable $O\in\mathcal{O}_V$ and $r>0$,
    \begin{align}
        \mathbb{P}_{\sigma}\big(|O-\langle O\rangle_{\sigma}\mathbb{I}|\ge r\big)\le  \frac{C|V|\,\|O\|_L^2}{r^2}\,.
    \end{align}
\end{enumerate}
\end{theorem}
Note that the Wasserstein distance and the Lipschitz constant are invariant under product unitaries. Thus, the same results hold for measuring the state in any product basis, not necessarily only the computational basis.

Although it might not be obvious from the outset, the inequalities in \autoref{thm:poicarconsequii} are known to imply no-go results for outputs of shallow quantum circuits~\cite{eldar2017local,bravyi_obstacles_2020,moosavian_limits_2021}. Consider for example the output distribution $\mu$ we obtain when measuring the GHZ-state in the computational basis, i.e. the all zeros or the all ones string. If we take $A$ to contain the all zeros string and $B$ to contain the all ones, we clearly have $\mu(A)=\mu(B)=0.5$ and $d_{H}(A,B)=n$. Thus, the GHZ state does not satisfy a $(2,\infty)$-Poincar\'{e} inequality with $C=\cO(1)$.

\subsection{Poincar\'{e} inequalities at the output of noisy circuits}\label{sec:poincare_shallow}

We will now bound the constant $C$ in various settings.
It turns out that noisy shallow circuits satisfy a $(2,\infty)$-Poincar\'{e} inequality:
\begin{prop}\label{shortdepth}
For any tensor product input state $\rho$, the output $\cN_V(\rho)$ satisfies a $(2,\infty)$-Poincar\'{e} inequality with constant 
$$C\le 4\,\Big(I_{\cN_V}^2+\frac{\max_\ell|E_\ell|}{|V|}\,\sum_{\ell=1}^L\,\max_{e\in E_\ell}\,I(e,L-\ell)^2\Big)$$
 where given a set $e\in E_\ell$ and $m\in\mathbb{N}$, $I(e,L-\ell)$ denotes the set of all vertices in $V$ in the light-come of the set $e$ for the circuit constituted of the last $L-\ell$ layers of $\mathcal{N}_V$.
 \end{prop}

The proof of this proposition is left to \autoref{appendixLipschitzcontrol}. 
When $\cN_V\equiv \mathcal{U}_V$ is noiseless, we get the following tightening of \autoref{shortdepth}:

\begin{prop}\label{noiselesslight-cone}
For any tensor product input state $\rho$, the output $\mathcal{U}_V(\rho)$ satisfies a $(2,\infty)$-Poincar\'{e} inequality with constant 
 \begin{align}\label{eq:noiselesslight-cone}
C\le {4\,I_{\cN_V}^2}\,.
 \end{align}
\end{prop}
Note that for any circuit the light-cone can grow at most exponentially in $L$.

\subsection{Poincar\'{e} inequality for continuous-time quantum processes}\label{sec:continuous}

We now consider the continuous-time setting, and restrict ourselves to a system whose interactions are modeled by a graph $G=(V,E)$ whose vertices $V$ correspond to a system of $|V|=n$ qudits, and denote by $D:= \max_{v\in V}\{v'|(v,v')\in E\}$ the maximum number of nearest neighbours to a vertex. In the (noiseless) continuous-time setting, one replaces the notion of a circuit by that of a local time-dependent Hamiltonian evolution:
 
\begin{defi}\label{contevol}
A (noiseless) continuous-time local quantum process is a unitary evolution $\{\mathcal{U}_V(t)\}_{t\ge 0}$ generated by the time-dependent Hamiltonian
\begin{align}
H(t)=\sum_{e\in E}\alpha_e(t)\,H_e\,,
\end{align}
where $H_e$ is a time-independent self-adjoint operator that acts non-trivially only on the edge $e\in E$ with norm $\|H_e\|_\infty \le \frac{1}{2}$. We also assume that $b:=\sup_{t,e}|\alpha_e(t)|<\infty$ independently of the size of the system. In what follows, for any subregion $A\subset  V$, we also denote the Hamiltonian restricted to $A$ by $H_A(t):=\sum_{e\subset A} \alpha_{e}(t)\,H_{e}$, and its corresponding unitary evolution by $\{\mathcal{U}_A(t)\}_{t\ge 0}$. 
\end{defi}

For continuous-time unitary evolutions, the concept of light-cone is formalised by the existence of a Lieb-Robinson bound \cite{lieb1972finite}. 
Since their introduction, Lieb-Robinson bounds have been extensively studied in various levels of generality for unitary \cite{nachtergaele2006lieb} as well as dissipative Markovian evolutions \cite{barthel2012quasilocality}. 
In what follows, we define a distance $\operatorname{dist}:E\times E\to \mathbb{R}_+$ on the edge set $E$ which for any two edges $e=(v_1,v_2)$ and $e'=(v_1',v_2')$ takes the value $\operatorname{dist}(e,e')=0$ if and only if $e=e'$, and otherwise is equal to the length of the shortest path connecting the sets of vertices $\{v_1,v_2\}$ and $\{v_1',v_2'\}$.  
Next, we denote by $S_e(k)$ the sphere around any edge $e\in E$ of radius $k$, i.e.
\begin{align}
    S_e(k):=\{e'\in E:\,\operatorname{dist}(e,e')=k\}\,.
\end{align}
Then, the set $E$ is said to be of spatial dimension $\delta$ if there is a constant $M>0$ such that for all $e\in E$, $|S_e(k)|\le M\,k^{\delta-1}$. The following result is taken from \cite[Theorem 2]{kliesch2014lieb} (see also \cite[Theorem 1]{moosavian_limits_2021} for a similar result).
\begin{theorem}[Lieb-Robinson bound]\label{thm:LR}
In the notations introduced above, for any subregions $A\subset B\subset V$ with $k_0:= \operatorname{dist}(A,V\backslash B)\ge 2\delta-1$, any state $\rho$ and $0\le t$:
\begin{align*}
    \Big\|\tr_{A^c}\Big( \mathcal{U}_V(t)(\rho)-\mathcal{U}_{B}(t)(\rho)\Big)\Big\|_1\le \frac{2M}{2D -1}\,k_0^{\delta-1}\,e^{v t-k_0}\,,
\end{align*}
where $v:=eb\,(2D-1)$ is the Lieb-Robinson velocity.
\end{theorem}

Next, we order the vertices $\{1,\cdots , n\}$, $n=|V|$, with their graph distance to an arbitrarily chosen vertex $v_0\equiv 1$, and denote the graph distance $\operatorname{dist}(\{1\},\{i,\cdots,n\})\equiv d(i)$. Then, in the notations of \autoref{notationsdefs}, we have by that for any $H\in\mathcal{O}_V$ (see \autoref{sec-liebrobin}),
\begin{align}
    &\|\mathcal{U}_V(t)^\dagger(H)\|_L\label{eqlipliebrob}\\
    &\quad \le \left(2(i_0-1)+\frac{4M}{2D-1}\sum_{i=i_0}^{|V|} d(i)^{\delta-1}\,e^{vt-d(i)}\right)\|H\|_L\,, \nonumber
\end{align}
where $i_0$ stands for the first vertex such that $d(i_0)\ge 2\delta-1$. By a reasoning that is identical to that leading to \autoref{shortdepth}, we have
\begin{prop}\label{conttimepoinca}
Let $\rho$ be a product input state. For any $t\ge 0$, the output state $\mathcal{U}_V(t)(\rho)$ satisfies a $(2,\infty)$-Poincar\'{e} inequality with constant
\begin{align*}
 C_t\le \,4\,\Big(2(i_0-1)+\frac{4M}{2D-1}\sum_{i=i_0}^{n} d(i)^{\delta-1}\,e^{vt-d(i)}\Big)^2\,.
\end{align*}
\end{prop}
For a simpler version of the bound found in \autoref{conttimepoinca}, we refer the reader to \autoref{thm:Chebyshevconcent}. The bounds obtained in \autoref{shortdepth}, \autoref{noiselesslight-cone} and \autoref{conttimepoinca} can be combined with \autoref{thm:poicarconsequ} (3) to get Chebyshev-type concentration bounds. This improves for instance over \cite[Theorem 2]{moosavian_limits_2021}, where the concentration bound was obtained only in the continuous-time Hamiltonian setting and for a specific 1-local observable measuring the Hamming weight. Moreover, the bound obtained in Equation \eqref{classisoperimetric} on the Hamming distance between two sets in terms of their probabilities in the state $\sigma$ is an improvement over the symmetric concentration inequality found in \cite[Corollary 43]{eldar2017local}, namely
\begin{equation}
    d_H(A,B)\le 4|V|^{1/2}I_{\mathcal{U}_V}^{3/2}\max\{\mu_\sigma(A)^{-1},\mu_\sigma(B)^{-1}\}\,,
\end{equation}
as well as its continuous-time analogue in \cite[Theorem 3]{moosavian_limits_2021}. In summary, the $(2,\infty)$-Poincar\'{e} inequality is a versatile tool that we use to derive the strongest concentration-type bounds for general short-time quantum evolutions currently available in a simple, basis-free manner.

\section{Limitations and concentration inequalities from noise}\label{sec:lci}
In \autoref{variancebound} we discussed how to use optimal transport methods to analyse the concentration profile of quantum circuits at small depths, even in the absence of noise. We now turn our attention to the case where the circuit is also subject to local noise and prove concentration inequalities for their outputs. 
As in the noiseless case, these can then be used to estimate the potential of noisy quantum circuits to outperform classical algorithms. However, unlike in~\autoref{thm:poicarconsequ}, we here obtain stronger Gaussian concentration inequalities.

 For this, we make use of the sandwiched R\'enyi divergences~\cite{MuellerLennert2013,Wilde2014} of order $\alpha\in(1,+\infty)$. For two states $\rho,\sigma$ such that the support of $\rho$ is included in the support of $\sigma$ they are defined as
 \begin{align*}
     D_\alpha(\rho\|\sigma)=\frac{1}{\alpha-1}\log \tr{\left[\left(\sigma^{\frac{1-\alpha}{2\alpha}}\rho\sigma^{\frac{1-\alpha}{2\alpha}}\right)^\alpha\right]}.
 \end{align*}
 We also consider the relative entropy we obtain by taking the limit $\alpha\to\infty$,
 \begin{align*}
     D_{\infty}(\rho\|\sigma)=\log(\|\sigma^{-\frac{1}{2}}\rho\sigma^{-\frac{1}{2}}\|_\infty)\,.
 \end{align*}
 In case the support of $\rho$ is not contained in that of $\sigma$, all the divergences above are defined to be $+\infty$.

We start from the assumption that the noise is driving the system to a quantum state $\sigma$ on $\mathcal{H}_V$ that satisfies a Gaussian concentration inequality of parameter $c>0$. That is, there is a constant $K$ such that for any $a>0$ and observable $O$:
\begin{align}\label{equ:gaussian_fixed_point}
    \mathbb{P}_\sigma\big(|O-\langle O\rangle_\sigma\mathbb{I}|\ge  a |V|  \big)\leq K\,e^{-\frac{ca^2|V|}{\|\sigma^{-\frac{1}{2}}O\sigma^{\frac{1}{2}}\|_{L}^2}}\,.
\end{align}
where the quantum Lipschitz constant of a non self-adjoint matrix $Z$ is defined as $\|Z\|_L:=\max\{\|\mathfrak{Re}(Z)\|_L,\|\mathfrak{Im}(Z)\|_L\}$. Note that inequalities of the form~\eqref{equ:gaussian_fixed_point} hold for product states~\cite{de2020quantum,Beigi2020,Rouze2019}, commuting high-temperature Gibbs states~\cite{Capel2020,depalma2021quantum} and in slightly weaker form for all high-temperature Gibbs states~\cite{Kuwahara2020} and gapped ground states on regular lattices \cite{Anshu2016}.
Moreover, in the case where $\sigma$ and $O$ commute, we clearly have $\|\sigma^{-\frac{1}{2}}O\sigma^{\frac{1}{2}}\|_{L}=\|O\|_{L}$.

We then have the following concentration result, proved in \autoref{lem:transfer_concent} of the Supplemental Material:
\begin{thm}\label{thm:concentration_renyi}
    Let $\sigma$ satisfy Eq.~\eqref{equ:gaussian_fixed_point}. Then for any state $\rho$ and $a>0$ and $\alpha>0$ we have:
\begin{align}\label{equ:concentration_transference}
    &\mathbb{P}_\rho\big(|O-\langle O\rangle_\sigma\mathbb{I}|\ge  a |V|  \big)\leq\nonumber \\&\operatorname{exp}\left(\frac{\alpha-1}{\alpha}\left(D_\alpha(\rho\|\sigma)-\frac{ca^2|V|}{\|\sigma^{-\frac{1}{2}}O\sigma^{\frac{1}{2}}\|_{L}^2}+\log(K)\right)\right).
\end{align}
\end{thm}
It immediately follows that if we have that for a noisy circuit and a value of $a$:
\begin{align}\label{equ:relative_entropy_density_cut_off}
    \frac{D_\alpha(\cN_V(\rho)\|\sigma)}{|V|}<\frac{ca^2}{\|\sigma^{-\frac{1}{2}}O\sigma^{\frac{1}{2}}\|_{L}^2}-\frac{\log(K)}{|V|}\,,
\end{align}
then the probability of observing an outcome outside of the interval $\langle O\rangle_\sigma\pm a|V|$ when measuring $\cN_V(\rho)$ is exponentially small in $|V|$. Thus, given a bound on $D_\alpha(\cN_V(\rho)\|\sigma)$, we can solve for $a$ in Equation~\eqref{equ:relative_entropy_density_cut_off} and establish $a$ such that the probability of observing outcomes outside of $\langle O\rangle_\sigma\pm a|V|$ is exponentially small. In \autoref{sec:examples} we discuss this more concretely to analyse the potential performance of QAOA under noise.

For now, let us discuss how to obtain the bounds on $D_\alpha(\cN_V(\rho)\|\sigma)$ to apply effectively~\autoref{thm:concentration_renyi}.
One straightforward way to derive such bounds is to resort to so-called strong data-processing inequalities (SDPI)~\cite{Carbone2015,Kastoryano2013,Hirche2020,Berta2019,Capel2020,Beigi2020,Olkiewicz1999,MllerHermes2018,gao2021complete}. A quantum channel $\cN$ with fixed-point $\sigma$ is said to satisfy an SDPI with constant $q_\alpha>0$ with respect to a fixed-point $\sigma$ and $D_{\alpha}$ if for all other states $\rho$ we have:
\begin{align}\label{equ:definition_SDPI}
D_\alpha(\cN(\rho)\|\sigma)\leq(1-q_\alpha)D_\alpha(\rho\|\sigma)\,.
\end{align}

Then, assuming that the noisy quantum circuit $\cN_V$ we wish to implement is of the form of \eqref{equ:localchannel} and each layer $\cN_{\ell}$ satisfies Eq.~\eqref{equ:definition_SDPI} for some constant $q_\alpha$, we show in \autoref{lem:dataprocessedtriangle} of the Supplemental Material that:
\begin{align}\label{equ:entropy_decay_main}
&D_\alpha(\cN_V(\rho)\|\sigma)\leq (1-q_\alpha)^LD_\alpha(\rho\|\sigma)\\
&\qquad \qquad \qquad +\sum\limits_{\ell=0}^L(1-q_\alpha)^{L-\ell}D_{\infty}\big(\bigotimes_{e\in E_\ell}\,\mathcal{N}_{\ell,e}(\sigma)\|\sigma\big).\nonumber
\end{align}
Thus, as long as the fixed point of the noise is left approximately invariant by the channels at the end of the circuit, Eq.~\eqref{equ:entropy_decay_main} implies that the relative entropy will decay as the depth increases. As we argue in \autoref{sec:qaoa_anneal}, this will be the case for both QAOA and annealing circuits for most one qubit noise models.
Furthermore, this will hold for any circuit whenever the fixed point of the noise is the maximally mixed state. 

It is also possible to derive similar inequalities for continuous-time evolutions with a time-dependent Hamiltonian $H_t$ and the noise given by some Lindbladian $\cL$. In that case, the assumption in Eq.~\eqref{equ:definition_SDPI} is replaced by
\begin{align}\label{equ:definition_SDPI_cont}
D_\alpha(e^{t\mathcal{L}}(\rho)\|\sigma)\leq e^{-r_{\alpha }t}D_\alpha(\rho\|\sigma)
\end{align}
for some constant $r_\alpha>0$.
In \autoref{thm:continuumlimit} we show the continuous-time version of Eq.~\eqref{equ:entropy_decay_main}.

To illustrate the power of the bound in Eq.~\eqref{equ:relative_entropy_density_cut_off}, let us analyse the case where $\cN_V$ consists in a concatenation of layers of unitary gates with layers of noise $\mathcal{N}=\otimes_{k=1}^{|V|}\cD_p$, where $\cD_p$ is a qubit depolarizing channel with depolarizing probability $p$. One can then show that Eq.~\eqref{equ:definition_SDPI} holds for $\alpha=2$ and $q_2=2p+p^2$~\cite[Sec. 3.3]{MllerHermes2018} and, thus, for any circuit of depth $L$ in this noise model:
\begin{align}
D_2\Big(\cN_V(\rho)\Big\|\frac{\mathbb{I}}{2^{|V|}}\Big)&\leq (1-p)^{2L}D_2\Big(\rho\Big\|\frac{\mathbb{I}}{2^{|V|}}\Big)\nonumber \\&\leq (1-p)^{2L}\,|V|\,.\label{equ:decay_depolarizing}
\end{align}
Moreover, the maximally mixed state satisfies Eq.~\eqref{equ:gaussian_fixed_point} with $c=K=1$~\cite{DePalma2021}. By combining Eq.~\eqref{equ:decay_depolarizing} with Eq.~\eqref{equ:relative_entropy_density_cut_off} we arrive at:
\begin{prop}\label{prop:exponentially_small_success}
    Let $H$ be a traceless $|V|$-qubit Hamiltonian and $\cN_{V}$ be a depth $L$ unitary circuit interspersed by $1-$qubit depolarizing noise with depolarizing probability $p$. Then for any initial state $\rho$ and $\epsilon>0$:
    \begin{align}\label{equ:gaussian_fixed_point_noise}
   & \mathbb{P}_{\mathcal{N}_V(\rho)}\big(|H|\ge  ((1-p)^{2L}+\epsilon)^{1/2} \|H\|_{L}|V|  \big)\\
    &\qquad\qquad \qquad \qquad \qquad \qquad \quad  \leq \mathrm{exp}\left(-\frac{\epsilon|V|}{2}\right).\nonumber
\end{align}
\end{prop}
Let us exemplify the power of Eq.~\eqref{equ:gaussian_fixed_point_noise}.
For an $H$ of practical interest, say $H$ an Ising Hamiltonian, efficient classical algorithms are known to find solutions whose energy is a constant fraction from the ground state energy~\cite{Bao2011}. 
That is, there exists an $a_c=\Omega(1)$ such that efficient classical algorithms can sample states $\rho$ that satisfy $\tr\left(\rho H\right)\leq -a_c|V|\|H\|_{L}$.

It then follows from Eq.~\eqref{equ:relative_entropy_density_cut_off} that at a constant depth $L>\log(a_c^{-1})/(2p)$, the probability of the noisy quantum circuit outperforming the classical algorithm is exponentially small in system size. 

Note that other results in the literature already showed that quantum advantage is already lost at constant depth for such problems~\cite{Franca2021a}. However, these results only showed bounds for the expectation value of the output of the circuit, whereas bounds like that in \autoref{prop:exponentially_small_success} provide concentration inequalities, a significantly stronger result. 
However, we do pay the price of having slightly worse constants for the depth at which advantage is lost compared to the results of~\cite{Franca2021a}. 
We will discuss concrete examples for the bounds we obtain on the depth in \autoref{sec:examples}.

Above we illustrated our concentration bounds for depolarizing noise only, as it corresponds to the simplest noise model that we can analyse. But our result can be generalized to all noise models that contract the relative entropy uniformly w.r.t. to a fixed point of full rank. However, this generalization comes at the expense of the bounds not being circuit-independent unless the noise is unital. As before, the first step to obtain concentration results is to control the decay of the relative entropy under the noise for R\'enyi divergences.
\begin{lem}[Lemma 1 of~\cite{Franca2021a}]\label{lem:dataprocessedtriangle_main}
    Let $\mathcal{N}:\cB(\cH_V)\to\cB(\cH_V)$ be a quantum channel with unique fixed point $\sigma>0$ that satisfies a strong data-processing inequality with constant $p_\alpha>0$ for some $\alpha>1$. That is,
    \begin{align}\label{equ:contraction_main}
    D_\alpha(\mathcal{N}(\rho)\|\sigma)\leq (1-p_\alpha)D_\alpha(\rho\|\sigma)
    \end{align}
    for all states $\rho$. Then for any other quantum channels $\Phi_1,\ldots,\Phi_m:\cB(\cH_V)\to\cB(\cH_V)$ we have:
    \begin{align}\label{equ:dataproccphi_main}
    &D_\alpha\Big(\prod_{t=1}^m (
    \Phi_t\circ  \mathcal{N})(\rho)\Big\|\sigma\Big)\leq (1-p_\alpha)^m D_\alpha(\rho\|\sigma)\nonumber\\&+\sum_{t=1}^{m} (1-p_\alpha)^{m-t} D_{\infty}( \Phi_{t}(\sigma)\|\sigma)\,.
    \end{align}
\end{lem}
We refer to \autoref{app:entropic_convergence} for a more detailed discussion of this result and \autoref{lem:dataprocessedtriangle} for a proof. In \autoref{sec:qaoa_anneal} we evaluate the expression in Eq.~\eqref{equ:dataproccphi_main} for the special case of QAOA circuits converging to diagonal product states. Furthermore, in \autoref{sec:regular_graphs} we discuss the performance of the resulting bounds for random graphs.

In the same appendix we also prove the continuous time version of the Lemma above that is relevant to quantum annealers, which now also state for completeness:
\begin{prop}\label{prop:entropy_decay_annealer_main}
    Let $\cL:\cB(\cH_V)\to\cB(\cH_V)$ be a Lindbladian with fixed point $\sigma_q$ defined as before with $q\geq\tfrac{1}{2}$. Suppose that for some $\alpha>1$ we have for all $t>0$ and initial states that there is a $r_\alpha>0$ such that:
    \begin{align}\label{equ:renyi_entropy_production1_main}
        D_\alpha(e^{t\cL}(\rho)\|\sigma)\leq e^{-r_\alpha t}D_\alpha(\rho\|\sigma).
    \end{align}
    Moreover, for functions $f,g:[0,1]\to\R$ and $T>0$ let $\cH_t:\cB(\cH_V)\to\cB(\cH_V)$  be given by $\cH_t(X)=i[X,f(t/T)H_X+g(t/T)H_I]$. Let $\mathcal{T}_t$ be the evolution of the system under the Lindbladian $\cS_{t}=\cL+\cH_t$ from time $0$ to $t\leq T$. Then for all states $\rho$:
    \begin{align}\label{equ:decay_ent_lind}
    &D_\alpha(\mathcal{T}_T(\rho)\|\sigma)\leq e^{-r_\alpha T}D_\alpha(\rho\|\sigma)+\\&2ne^{-r_\alpha T}\left(\sqrt{\frac{q}{1-q}}-\sqrt{\frac{1-q}{q}}\right)\int\limits_{0}^T  e^{r_\alpha t}|f(t/T)|\,dt\,.
    \end{align}    
\end{prop}

Note that the expression in Eq.~\eqref{equ:dataproccphi_main} will converge to $0$ as long as $\Phi_t(\sigma)\simeq \sigma$ for $t$ close to $T$. 
As we will argue in more detail in \autoref{app:entropic_convergence}, this is expected to be satisfied for good QAOA circuits. Furthermore, we explicitely evaluate the bound in Eq.~\eqref{equ:dataproccphi_main} in terms of the parameters of the QAOA circuit in \autoref{cor:QAOA_rel_ent} or for a given annealing schedule. These results can then be combined with \autoref{thm:concentration_renyi} to understand the concentration properties of the output. The same holds in principle for Eq.~\eqref{equ:decay_ent_lind}, where this can be visualized more easily: As long as the function $f$ satisfies $f(1)=0$, the second term in Eq.~\eqref{equ:decay_ent_lind} will converge to $0$.

We illustrate this concretely in the case of noisy annealers with a linear schedule in \autoref{prop:noisy_annealer_concent_main}.

\section{Limitations of error mitigated noisy VQAs: concentration bounds}\label{mitigationsec}

A possible criticism of bounds like that of \autoref{prop:exponentially_small_success} is that they do not take error-mitigation techniques~\cite{Temme2017,Endo2018,McArdle2019,Endo2021} into account. 
Although there does not seem to be a widely accepted definition of what error mitigation entails, the overarching goal of such protocols is to extract information about noiseless circuits by sampling from noisy ones. Such proposals are expected to be useful before the advent of fault tolerance to reduce the level of noise present in the data outputted by NISQ devices.

The majority of existing mitigation protocols require a significant overhead in the number of samples to extract the noiseless signal from noisy ones, potentially making error mitigation prohibitively expensive. Thus, one of the main questions regarding the viability of error mitigation strategies is the scaling of the sampling overhead they require in terms of number of qubits, depth and error rate.

There already exist some results in the literature discussing limitations of error mitigation, such as~\cite{Takagi2021,Wang2021}. They show that certain error mitigation protocols require a sampling overhead that is exponential in system size at linear circuit depth. Our results in the next sections suggest that at significantly lower depths it is already difficult to extract information about the noiseless output state, while also providing concentration bounds for error-mitigated circuits. 

In what follows, we will distinguish \emph{sampling} and \emph{weak} error mitigation. To the best of our knowledge this distinction has not been made before in the literature. But in analogy with the terminology for the simulation of quantum circuits, we will call an error mitigation strategy a sampling protocol if it allows us to approximately sample from the output of a noiseless circuit. In contrast, weak error mitigation techniques only allow for approximating expectation values of the outputs of noiseless circuits. Note that the latter is a weaker condition.

\subsection{Sampling error mitigation and the effect of error mitigation on classical optimization problems}\label{sec:virtual_cooling}

We start by discussing the effect of noise on known sampling error mitigation procedures.
We believe that these are particularly relevant for classical combinatorial optimization problems. This is because for such problems one is often not necessarily interested in estimating the ground state energy, but rather in obtaining a string of low energy that corresponds to a good solution. And for this it would be necessary to obtain a sample.

To the best of our knowledge, the only error mitigation technique that allows for sampling from the noiseless state is virtual distillation or cooling~\cite{Huggins2021,Koczor2021}.
Going into the details of this procedure is beyond the scope of this manuscript. It suffices to say that it takes as an input $k$ copies of the output of a noisy quantum circuit and aims at preparing the state $\rho^k/\tr\left[\rho^k\right]$. Under some assumptions, one can then show that this state has an exponentially in $k$ larger overlap with the output of the noiseless circuit. However, as this is clearly not a linear transformation, it can only be implemented stochastically. The success probability of the transformation is given $\tr\left[\rho^k\right]\leq\tr\left[\rho^2\right]$. 
As before, for simplicity we will state our no-go results for the case of local depolarizing noise and leave the proof and the more general case to \autoref{sec:qaoa_anneal}.
We then have:
\begin{prop}\label{prop:higher_moment}
    Let $\cN_{V}$ be a depth $L$ unitary circuit interspersed by $1-$qubit depolarizing noise with depolarizing probability $p$. Then for any initial state $\rho$ and $k\geq 2$, the probability that virtual cooling or distillation succeeds is bounded by:
    \begin{align}\label{equ:bound_purity}
        \tr\left[\cN_{V}(\rho)^k\right]&\leq\tr\left[\cN_{V}(\rho)^2\right]\nonumber\\ &\le\operatorname{exp}(-\log(2)(1-(1-p)^{2L})n)
    \end{align}

\end{prop}
The proof of \autoref{prop:higher_moment} can be found in \autoref{prop:purity_vanishes}.
Thus, we conclude from Eq.~\eqref{equ:bound_purity} that unless the local noise rate is $p=\cO(n^{-1})$, virtual distillation protocols will require an exponential in system size number of samples to be successful even after one layer of the circuit. We remark that our results essentially imply the same conclusions for general local, unital noise. For nonunital noise driving the system to a product state we obtain the following statement:
\begin{lem}\label{lem:bound_purity_noise_main}
    Let $\tau_q=q\ketbra{0}+(1-q)\ketbra{1}$ and assume w.l.o.g. that $q\leq\tfrac{1}{2}$. Then for for any state $\rho\in\mathcal{S}_V$ with $n=|V|$ such that 
    \begin{align}\label{equ:rho_sigma_purity_main}
    D_2(\rho\|\tau_q^{\otimes n})\leq(1-\epsilon-\log(2(1-q)))n\,,
    \end{align}
     we have that the probability that virtual cooling or distillation succeeds is bounded by $2^{-\epsilon n}$.
\end{lem}
That is, for more general fixed points the virtual cooling or distilliation will succeed with exponential small probablity if the realative entropy has decayed by a factor of $\log(2q)$. As it was the case with the concentration bounds, we see that our bounds become weaker as the fixed point becomes purer. Furthermore, it is also possible to immediately apply the results derived in \autoref{app:entropic_convergence} to estimate when the entropy has contracted enough such that the success probability becomes exponentially small.

\subsection{Weak error mitigation with regular estimators}\label{sec:weakerror}
We will now see that the techniques of the last sections also readily apply to weak error mitigation techniques that are regular in a sense that will be made precise later. 
To the best of our knowledge, all weak error mitigation techniques have the following basic building blocks and parts:
\begin{enumerate}
\item Take the outcome of $m$ (noisy) quantum circuits $\mathcal{E}_1,\cdots,\mathcal{E}_m$ with initial states $\rho_1,\cdots, \rho_m$.
\item Add auxiliary qubits and perform a collective noisy circuit $\Phi$ on the output of the $m$ circuits.\label{enu:ancilla}
\item Perform a measurement on the $m$ systems.\label{enu:measurement}
\item Postprocess the outcomes of the measurements and output an estimate\label{enu:postprocessing}.
\end{enumerate}
This is illustrated in \autoref{fig:diagrammitigation}.
It is easy to see that points ~\eqref{enu:ancilla}, \eqref{enu:measurement} and \eqref{enu:postprocessing} can all be collectively modelled by applying a global projective measurement $\mathcal{M}:=\{M_s\}_{s\in\mathcal{S}}$ on the state $\bigotimes_{i=1}^m\mathcal{E}_i(\rho_i)\otimes \ketbra{0}^{\otimes k}$, where we assume we have access to $k$ auxiliary systems. Here the PVM is indexed from some classical sample space $\mathcal{S}$, followed by a classical procedure mapping each measured output $s\in\mathcal{S}$ to a real value $f(s)$ through a function $f:\mathcal{S}\to\R$. The hope is then that $f(s)$ provides a good estimate for some property of the noiseless circuit.

Equivalently, we are interested in the probabilistic properties of the observable
\begin{align}\label{equ:error-mitigation-obs}
    X:=\sum_{s\in\mathcal{S}}f(s)\,\tr_{A}\left({I_{S}\otimes \ketbra{0}^{\otimes k}M_s}\right)
\end{align}
in the output state $\rho_{\operatorname{out}}:=\bigotimes_{i=1}^m\mathcal{E}_i(\rho_i)$ of the original noisy circuit, where we have traced out the auxiliary systems used in the mitigation process.

\begin{figure}[H]
    \centering
    \includegraphics[width=\columnwidth]{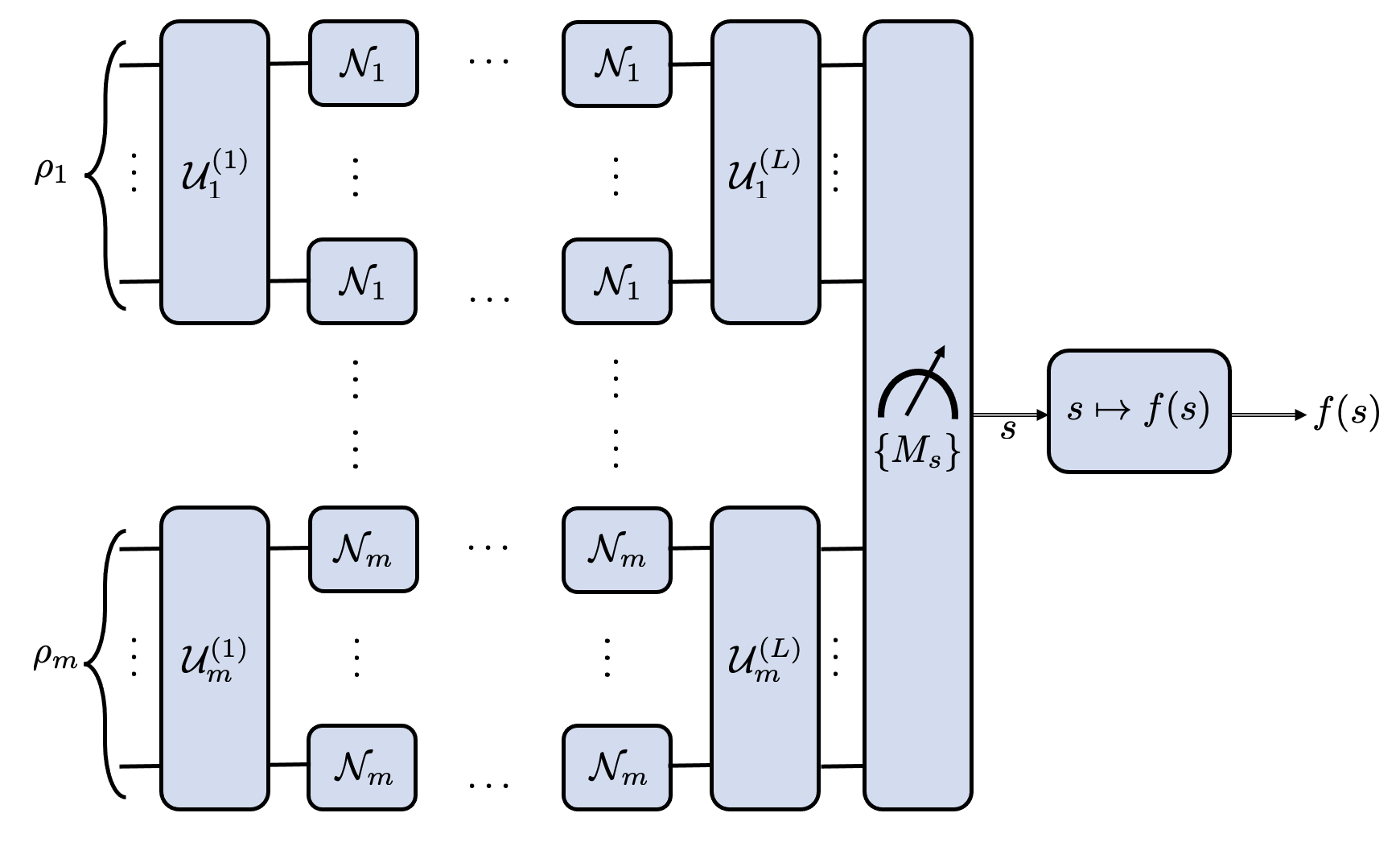}
    \caption{Schematic of an error mitigation protocol}
    \label{fig:diagrammitigation}
\end{figure}

In order to obtain concentration inequalities for error mitigation protocols, we will impose a bit more structure on the estimators. To make our motivation for our further assumptions clear, we will use as our guiding example the most naive of all error mitigation protocols for an optimization task: sampling $m$ times from the quantum device, evaluating the energy of each outcome and outputting the minimum. I.e., just repeating the experiment often enough. First, we will assume that the PVM is indexed by labels $s\in\R^{m}$. In the case of the minimum strategy discussed before, the individual entries of this vector would correspond to the energy we observed on each one of the $m$ copies. Furthermore, we will assume that $f$ is $L_f$-Lipschitz w.r.t. the $\ell_\infty$ norm on $\R^m$, i.e.:
\begin{align*}
    \sup\limits_{s,s'\in\R^m}\frac{|f(s)-f(s')|}{\|s-s'\|_{\ell_\infty}}\leq L_f.
\end{align*}
In the case of the minimum strategy, $f$ would correspond to the minimum in $\R^m$, for which we have $L_f=1$. Let us justify the assumption that $f$ is Lipschitz by looking in a bit more detail into the case of the POVM measuring copies independently. In that case, $f$ being Lipschitz w.r.t. $\ell_\infty$ corresponds to requiring that the error-mitigated estimate should not depend too strongly on any individual sample, a robustness condition that is desirable in the presence of noise. Finally, we will assume that the error mitigation procedure concentrates when given trivial, product states:
\begin{align}\label{equ:collective_concentration}
    &\mathbb{P}_{\sigma^{\otimes m}}(\|s-\mathbb{E}(s)\|_{\ell_\infty}\geq rn)\leq K(m)\textrm{exp}\left(-\frac{cr^2n}{\ell_0^2}\right)
\end{align}
for some function $K(m)$.

Let us discuss this assumption once again in the case of taking the minimum of measuring the energy of a Hamiltonian $H$ $m$ times. In that case, each measurement satisfies Gaussian concentration for some $c$ and $\ell_0=\|\sigma^{-\frac{1}{2}}H\sigma^{\frac{1}{2}}\|_L$. Thus, it follows from a union bound that for taking independent measurements, Eq.~\eqref{equ:collective_concentration} holds with $K(m)=m$.

Now that we have formulated the error mitigation protocol in this way, we can immediately apply the same reasoning as in \autoref{prop:exponentially_small_success} to understand the concentration properties of the error mitigation procedure. We get:

\begin{theorem}\label{thm:theoremmain}
For an error mitigation observable $X$ as in Eq.~\eqref{equ:error-mitigation-obs} assume that for a given state $\sigma$ Eq.\eqref{equ:collective_concentration} holds for some function $K(m)$.
Furthermore, assume that for $r,\epsilon>0$ given we have for all $1\leq i\leq m$ that $D_2(\mathcal{E}_i(\rho)\|\sigma)\leq \tfrac{c(r^2-\epsilon)n}{ml_0^2}$. Then:
\begin{align}\label{equ:concentration_mitigation}
    \mathbb{P}_{\rho_{\operatorname{out}}}(|X-f(\mathbb{E}_{\sigma^{\otimes m}}(s))\mathbb{I}|> rL_fn)\le \mathrm{exp}(-\frac{c\epsilon n}{\ell_0^2})\,.
\end{align}

\end{theorem}
We leave the proof of~\autoref{thm:theoremmain} to~\autoref{sec:concentracion_error}.
We see that the amount by which the R\'enyi entropy has to decrease to ensure we are in the regime where we obtain concentration from Eq.~\eqref{equ:concentration_mitigation} is connected to the Lipschitz constant of $X$ and the number of copies $m$. For instance, under local depolarizing noise with depolarizing probability $p$, this happens at depth $\cO(p^{-1}\log(ml_0))$.

One way of interpreting the bound in Eq.~\eqref{equ:concentration_mitigation} is that the probability that the estimate we obtain from the output of the error mitigation algorithm with input given by the noisy states to that with the fixed-point of the noise as input is exponentially small.
Thus, the noisy outputs were useless: we could have just sampled from the product state $\sigma^{\otimes m}$ instead and observed similar outcomes.

However, it might be hard to control the Lipschitz constant $L_f$ in general scenarios. Moreover, many mitigation protocols in the literature~\cite{Temme2017,Endo2018,McArdle2019,Endo2021} involve estimating the mean of random variables that take exponentially large values. Thus, their Lipschitz constant will typically also be exponentially large, constraining the applicability of \autoref{thm:theoremmain}.

\section{Example: finding the ground state of Ising Hamiltonians in the NISQ era}\label{sec:examples}
Given a matrix $A\in\R^{n\times n}$ and a vector $b\in\R^n$ we define the Hamiltonian
\begin{align}\label{equ:Ising_model}
H_I=-\sum_{i,j=1}^nA_{i,j}Z_iZ_j-\sum\limits_
{i=1}^nb_iZ_i.
\end{align}
It is well-known how to formulate various NP-complete combinatorial optimization problems as finding a string that minimizes the energy of $H_I$. This has motivated the pursuit of NISQ algorithms for this task, including the quantum approximate optimization algorithm~\cite{Farhi2014} (QAOA) or the closely related quantum annealing algorithm.

Let us briefly describe the QAOA algorithm. Given a $P\in\N$ and vectors of parameters $\gamma,\tau\in\R^P$, the QAOA unitary is given by
\begin{equation}\label{eq:QAOA}
V_{\gamma, \beta}=\prod\limits_{k=1}^P e^{i\beta_k H_X}e^{i\gamma_k H_I},
\end{equation}
where $H_X=-\sum_{i=1}^nX_i$. The hope of QAOA is that by optimizing over the parameters $\gamma,\beta$, measuring $V_{\gamma, \beta}\ket{+}^{\otimes n}$ in the computational basis will yield low energy strings for the Hamiltonian in Eq.~\eqref{equ:Ising_model} even for moderate values of $P$. In what follows we will distinguish the depth of the QAOA Ansatz (denoted by $P$) from the physical depth of the circuit being implemented in the device (denoted by $L$).

In recent years, several works have identified limitations on the performance of constant depth circuits in outperforming classical algorithms for this problem~\cite{bravyi_obstacles_2020,farhi_quantum_2020}, even in the absence of noise. These results were then later extended to short-time quantum annealing~\cite{moosavian_limits_2021}.

Taking the noise into consideration, recent works have shown that QAOA is outperformed by efficient classical algorithms at a depth that is proportional to the local noise rate~\cite{Franca2021a}. 
However, those works only considered the expected value of the output string. Considering that the goal of QAOA is to obtain one low-energy string, to completely discard exponential advantages of QAOA and other related algorithms at a depth that only depends on local noise rates, it is important to also obtain concentration inequalities for the outputs.

As mentioned before, \autoref{prop:exponentially_small_success} already allows us to conclude that quantum advantage will be lost against classical algorithms at constant depth.
With the techniques presented in this work, it is also straightforward to obtain concentration bounds for concrete instances. Indeed, given that a classical algorithm found a string with given energy $-a_Cn$, we can easily bound the depth at which the bound in \autoref{prop:exponentially_small_success} kicks in and the quantum device is exponentially unlikely to yield a better result.

\subsection{Max-Cut} \label{sec:Max-Cut}
In this subsection, we analyze the performances of quantum circuits for the Max-Cut problem.
Let $G=(V,E)$ be a graph.
The \emph{cut} of a bipartition of $V$ is the number of edges that connect the two parts.
The Max-Cut problem consists in finding the maximum cut of $G$, which we denote with $C_{\max}$.
The best classical algorithm for Max-Cut is due to Goemans and Williamson \cite{goemans1995improved} and can obtain a string whose cut is at least $0.878\,C_{\max}$.
As in \cite{bravyi_obstacles_2020}, we consider circuits that commute with $\sigma_x^{\otimes n}$, which include the QAOA circuit.
We prove that the algorithm by Goemans and Williamson cannot be outperformed by:
\begin{itemize}
    \item Noiseless circuits with shallow depth (\autoref{thm:noiselessMax-Cut});
    \item Noisy circuits with any depth (\autoref{thm:noisyMax-Cut}).
\end{itemize}

We assume that $G$ is bipartite, \emph{i.e.}, $C_{\max} = |E|$, and is regular with degree $D$, \emph{i.e.}, each vertex belongs to exactly $D$ edges.
Without loss of generality, we assume $V=[n]$.
We associate to each bipartition $V=V_0\cup V_1$, $V_0\cap V_1=\emptyset$ the bit string $x\in\{0,1\}^n$ such that $x_i=j$ if $i\in V_j$.
We denote with $C(x)$ the cut of such bipartition.
We also assume that $G$ satisfies
\begin{equation}\label{eq:hypC}
    C(x) \ge h\min\{|x|,n-|x|\}\,,\qquad h = \frac{D}{2}-\sqrt{D-1}
\end{equation}
for any $x\in\{0,1\}^n$, where $|x|$ denotes the Hamming weight of $x$, \emph{i.e.}, the number of components of $x$ that are equal to $1$.
For any $D\ge3$, Ramanujan expander graphs constitute an example of graphs with such property \cite{morgenstern1994existence,marcus2013interlacing,marcus2018interlacing}.
Moreover, random $D$-regular bipartite graphs approach the bound \eqref{eq:hypC} with high probability \cite{friedman2008proof}.

The Max-Cut problem for $G$ is equivalent to maximizing the $n$-qubit Hamiltonian
\begin{equation}\label{eq:MCH}
    H = \sum_{x\in\{0,1\}^n} C(x)\,|x\rangle\langle x| = \frac{1}{2}\sum_{(j,k)\in E}\left(\mathbb{I} - \sigma_z^j\,\sigma_z^k\right)\,,
\end{equation}
where for any $j\in[n]$, $\sigma_z^j$ is the Pauli $Z$ matrix acting on the qubit $j$.
\begin{thm}[noiseless Max-Cut]\label{thm:noiselessMax-Cut}
Let $G$ be a regular bipartite graph with $n$ vertices satisfying \eqref{eq:hypC}, and let $H$ be the associated Max-Cut Hamiltonian \eqref{eq:MCH}.
Let $\rho$ be the output of a noiseless quantum circuit as in \autoref{defi:circuit} made by $L$ layers, where each layer consists of a set of unitary gates acting on mutually disjoint couples of qubits.
We assume that the input state of the circuit and each unitary gate commute with $\sigma_x^{\otimes n}$.
Then, if
\begin{equation}\label{eq:appr}
    \mathrm{Tr}\left[\rho\,H\right] \ge C_{\max}\left(\frac{5}{6} + \frac{\sqrt{D-1}}{3D}\right) \,,
\end{equation}
we must have
\begin{equation}\label{eq:boundL}
L\ge \frac{1}{2}\log_2\frac{n}{576}\,.
\end{equation}
Furthermore, if $\rho$ is generated by the QAOA circuit \eqref{eq:QAOA} with depth $P$, we must have
\begin{equation}\label{eq:boundP}
P\ge \frac{1}{2\log\left(D+1\right)}\log\frac{n}{576}\,.
\end{equation}
\end{thm}

\begin{rem}
For any $D\ge55$ we have
\begin{equation}
    \frac{5}{6} + \frac{\sqrt{D-1}}{3D} < 0.878\,,
\end{equation}
therefore any quantum algorithm that outperforms the algorithm by Goemans and Williamson must generate a state satisfying \eqref{eq:appr}.
\end{rem}

\begin{rem}
Under the same hypotheses of \autoref{thm:noiselessMax-Cut}, Ref. \cite[Corollary 1]{bravyi_obstacles_2020} proves that
\begin{equation}\label{eq:BK}
    P \ge \frac{1}{3\left(D+1\right)}\log_2\frac{n}{4096}\,.
\end{equation}
Our result \eqref{eq:boundP} provides an exponential improvement over \eqref{eq:BK} with respect to $D$.
Already for $D=55$, the right-hand side of \eqref{eq:BK} is larger than $1$ only for $n=\Omega(10^{54})$, while the right-hand side of \eqref{eq:boundP} is larger than $1$ already for $n = \Omega(10^6)$.
\end{rem}

\begin{proof}
\emph{Circuit made of two-qubit gates:}
From \autoref{noiselesslight-cone}, $\rho$ satisfies a $(2,\infty)$ Poincar\'e inequality with constant
\begin{equation}\label{eq:CLn}
    C \le 2^{2L+2}\,.
\end{equation}
Let
\begin{align}
    A &= \left\{x\in\{0,1\}^n:d_H(x,x_{\mathrm{opt}}) \le \frac{n}{3}\right\}\,,\nonumber\\
    B &= \left\{x\in\{0,1\}^n:d_H(x,\bar{x}_{\mathrm{opt}}) \le \frac{n}{3}\right\}\,,
\end{align}
and let $X$ be the random outcome obtained measuring $\rho$ in the computational basis.
\autoref{prop:symmetry} of \autoref{app:MC} implies
\begin{equation}
    \mathbb{P}(X\in A) = \mathbb{P}(X\in B) \ge \frac{1}{4}\,,
\end{equation}
and \eqref{classisoperimetric} of \autoref{thm:poicarconsequ} together with \eqref{eq:CLn} imply
\begin{equation} \label{eq:proofsym}
    \frac{n}{3} \le d_H(A,B) \le 2^{L+3}\sqrt{n}\,.
\end{equation}
The claim \eqref{eq:boundL} follows.

\emph{QAOA circuit:}
From \autoref{noiselesslight-cone}, $\rho$ satisfies a $(2,\infty)$-Poincar\'e inequality with constant
\begin{equation}
    C \le 4\left(D+1\right)^{2P}\,.
\end{equation}
Proceeding as in the previous case we get
\begin{equation}
    \frac{n}{3} \le d_H(A,B) \le 8\left(D+1\right)^P\sqrt{n}\,,
\end{equation}
and the claim \eqref{eq:boundP} follows.
\end{proof}

\begin{thm}[noisy Max-Cut]\label{thm:noisyMax-Cut}
Under the same hypotheses of \autoref{thm:noiselessMax-Cut}, let each layer of the circuit be followed by depolarizing noise with depolarizing probability $p$ applied to each qubit.
Then,
\begin{equation}\label{eq:boundn}
n \le 3\cdot2^{\frac{2}{p}+8}\,.
\end{equation}
For $p=0.1$, \eqref{eq:boundn} gives $n\le 8\cdot 10^8$.
\end{thm}
\begin{proof}
If $n\le3072$, the bound \eqref{eq:boundn} is empty.
We can then assume $n\ge 3072$.
Proceeding as in the proof of \autoref{thm:noiselessMax-Cut} and employing \autoref{shortdepth} in place of \autoref{noiselesslight-cone} we get
\begin{equation}\label{eq:Llower}
n \le 3\cdot2^{2L+8}\,.
\end{equation}
Let us consider the following operator associated to the Hamming distance from $x_{\mathrm{opt}}$:
\begin{equation}
    K = \sum_{x\in\{0,1\}^n}d_H(x,x_{\mathrm{opt}})\,|x\rangle\langle x| - \frac{n}{2}\,\mathbb{I}\,.
\end{equation}
We have $\mathrm{Tr}\,K = 0$ and $\|K\|_L = 1$, therefore \autoref{prop:exponentially_small_success} implies that for any $\epsilon>0$, upon measuring $K$ on $\rho$ we have
\begin{equation}\label{eq:K}
    \mathbb{P}\left(|K|\ge \left(\epsilon + \left(1-p\right)^{2L}\right)n\right) \le 2\,e^{-\frac{\epsilon n}{2}}\,.
\end{equation}
\autoref{prop:symmetry} implies
\begin{equation}\label{eq:PK}
    \mathbb{P}\left(\left|\frac{1}{n}\,d_H(X,x_{\mathrm{opt}}) - \frac{1}{2}\right|\ge \frac{1}{6}\right) \ge \frac{1}{2}\,,
\end{equation}
and choosing in \eqref{eq:K}
\begin{equation}
    \epsilon = \frac{1}{6} - \left(1-p\right)^{2L}
\end{equation}
we get
\begin{equation}
    \left(1-p\right)^{2L} \ge \frac{1}{6} - \frac{4}{n}\ln2 \ge e^{-2}\,,
\end{equation}
hence
\begin{equation}\label{eq:Lupper}
    L \le -\frac{1}{\ln\left(1-p\right)} \le \frac{1}{p}\,.
\end{equation}
The claim follows combining \eqref{eq:Llower} and \eqref{eq:Lupper}.
\end{proof}

\subsection{Short-time evolution of local Hamiltonians}
Quantum annealing constitutes another family of heuristic algorithms to solve optimization problems. Similar to the variational algorithms discussed earlier, the goal in quantum annealing is to find the lowest energy of a classical Hamiltonian that encodes the optimization problem. To find the lowest energy of the optimization Hamiltonian, we can start from a local Hamiltonian whose ground state is easy to prepare, for example  $-\sum_i X_i$, and continuously change the Hamiltonian to the desired optimization Hamiltonian $H_I$:
\begin{equation} \label{eq:anneal}
  H(t)= -a(t) \sum_i X_i + b(t) H_I \, ,
\end{equation}
where $a(0)=b(T)=1,\, a(T)=b(0)=0$ and $T$ is the final evolution time. The adiabatic theorem~\cite{Jansen_2007} guarantees that if we start from the ground state of the initial Hamiltonian and evolve the system slowly enough, the final state would be close to the ground state of the optimization Hamiltonian, which can be found by measurement in the computational basis at the final time. Since noise restricts the total time that coherence in the system is preserved, understanding the limitations of short-time evolution of local Hamiltonians seems crucial. The presented $(2,\infty)$-Poincar\'{e} inequality provides bounds on the performance of short-time quantum annealers.

\begin{prop}[Short-time evolution of local Hamiltonians]\label{thm:Chebyshevconcent}
Let $\sigma$ be the quantum state generated by evolving a product state with a continuous-time local
quantum process as in \autoref{sec:continuous} for time $t\ge0$.
Let $\mu_\sigma$ be the probability distribution of the outcome of the measurement in the computational basis performed on $\sigma$.
Then, for any $A,\,B\subseteq\{0,1\}^V$ we have
\begin{align*}
    d_H(A,B)\le \sqrt{c_0+c_1  e^{vt}} \sqrt{|V|}\,\big(\mu_\sigma(A)^{-\frac{1}{2}}+\mu_\sigma(B)^{-\frac{1}{2}}\big)\, ,
\end{align*}
where $d_H$ denotes the Hamming distance, $v=eb\,(2D-1)$, $D$ is the maximum degree of the interaction graph, $b$ is the maximum interaction strength and
\begin{align}
c_0=64 M \delta^\delta , \,  c_1=\frac{64 M}{2D-1} Li_{-2 (\delta-1)}(e^{-1}) \, ,
\end{align}
where $Li_{s}(z)$ is the polylogarithm function of order $s$ and argument $z$ and $\delta$ is spatial dimension of the interaction graph.
\end{prop}

\begin{rem}
Crucially, both $c_0$ and $c_1$ are independent of the number of qubits.
\end{rem}

\begin{proof}
We start by deriving an upper bound on  $C_t$ of  \autoref{thm:poicarconsequ}. We note that by the definition of $i_0$,  we have  $2\delta-1 \leq d(i_0)\leq 2\delta$, and therefore using $|S_e(k)|\le M\,k^{\delta-1}$ we have $i_0 \leq \sum_{k=0}^{2\delta} |S_e(k)|\leq 2\delta M\,\delta^{\delta-1}+1=2 M\,\delta^{\delta}+1$. Also, we have
\begin{align}
\sum_{i=i_0}^{n} d(i)^{\delta-1}\,e^{-d(i)} \leq \sum_{k=2\delta-1}^{n}  |S_e(k)| k^{\delta-1}\,e^{-k} \nonumber %
\\ \leq  M \sum_{k=2\delta-1}^{n}  k^{2 (\delta-1)} \,e^{-k} \leq M Li_{-2 (\delta-1)}(e^{-1}) \, .
\end{align}
Putting these two bounds together, we have
\begin{align}  \label{eq:boundCt}
 C_t \le (c_0+c_1 e^{v t})^2 \, .
\end{align}
The claim follows by applying \autoref{thm:poicarconsequ}.
\end{proof}

Considering the example of generating a generalized GHZ state, where $d_H(A,B)=n$ and $\mu_\sigma(A)=\mu_\sigma(B)=\frac{1}{2}$, we have
\begin{align}
    \frac{1}{v}\log\left(\frac{n}{8 c_1} -\frac{c_0}{c_1}\right)\le t \, ,
\end{align}
and, therefore, at least $O(\log(n))$ time is required to generate generalized GHZ states using local Hamiltonians. Note that this bound also provides a minimum time required by local Hamiltonians to simulate unitaries that are capable of generating generalized GHZ states starting from product states, such as n-qubit fan-out gates.

The short-time evolution of local Hamiltonians also limits their performance to solve Max-Cut problem discussed in  \autoref{sec:Max-Cut}. Note that both the initial state and the annealing Hamiltonian of \eqref{eq:anneal} with the final Hamiltonian $H_I$ corresponding to the Max-Cut problem commute with $\sigma_x^{\otimes n}$, and therefore the techniques of \autoref{thm:noiselessMax-Cut}  directly lead to a proof for limitation of short-time evolution of local Hamiltonian for the optimization task.

\begin{prop}\label{prop:annealing} Consider the Max-Cut problem Hamiltonian $H_I$ as discussed in \autoref{thm:noiselessMax-Cut}, and the corresponding annealing Hamiltonian in the form of \eqref{eq:anneal}. Let $\rho$ be evolved states after time $T$. Then, if
\begin{equation}
    \mathrm{Tr}\left[\rho\,H\right] \ge C_{\max}\left(\frac{5}{6} + \frac{\sqrt{D-1}}{3D}\right) \,,
\end{equation}
we must have
\begin{equation}
T \geq \frac{1}{v}\ln\left(\frac{\sqrt{n}}{12 c_1} - \frac{c_0}{c_1}\right)
\end{equation}
 \end{prop}
\begin{proof} From  \eqref{eq:proofsym} and \eqref{classisoperimetric} of \autoref{thm:poicarconsequ} we have 
\begin{equation}
    \frac{n}{3} \le d_H(A,B) \le 4 \sqrt{C_t}  \sqrt{n}\,.
\end{equation} 
which can be combined with  \eqref{eq:boundCt} to get
\begin{equation}
    \frac{\sqrt{n}}{12} \le c_0+c_1 e^{vT}
    \,.
\end{equation} 
The claim follows.
\end{proof}

\subsection{Noisy QAOA beyond unital noise}\label{subsec:example_QAOA}
In this subsection we discuss the performance of our bounds for QAOA beyond the case of unital noise. As mentioned before, if the noise is not unital our bounds on the relative entropy decay are not independent of the circuit being implemented. Thus, we need to pick a promising family of QAOA parameters to apply our results.

A natural candidate of instances to analyse is Max-Cut on random regular graphs of high girth. This is because in~\cite{Basso2021a} the authors derive the optimal parameters for QAOA for such graphs in the large $n$ limit for up to $17$ layers. Furthermore, they show that these QAOA circuits achieve an expected value for the cut that is higher than what known provably efficient classical algorithms achieve.
Although these parameters are only optimal in the absence of noise, we analyse their performance in the presence of non-unital noise driving the system to the classical state $\tau_q^{\otimes n}$ with $\tau_q=q\ketbra{0}+(1-q)\ketbra{1}$.

As explained in \autoref{sec:regular_graphs}, we show that as long as the output $\rho$ of a noisy QAOA circuit satisfies 
\begin{align}\label{equ:threshold_graph_main}
    D_2(\rho\|\tau_q^{\otimes n})< \frac{((1-2q)^2\frac{D}{2}+\frac{2}{\pi}\sqrt{D})^2}{2D^2}n,
\end{align}
for a $D$-regular graph, the probability that the noisy circuit outperforms classical methods is exponentially small. When this is achieved in terms of the contraction coefficient is displayed in \autoref{fig:plot_noisyqaoa}.

\begin{figure}[H]
    \centering
    \includegraphics[width=\columnwidth]{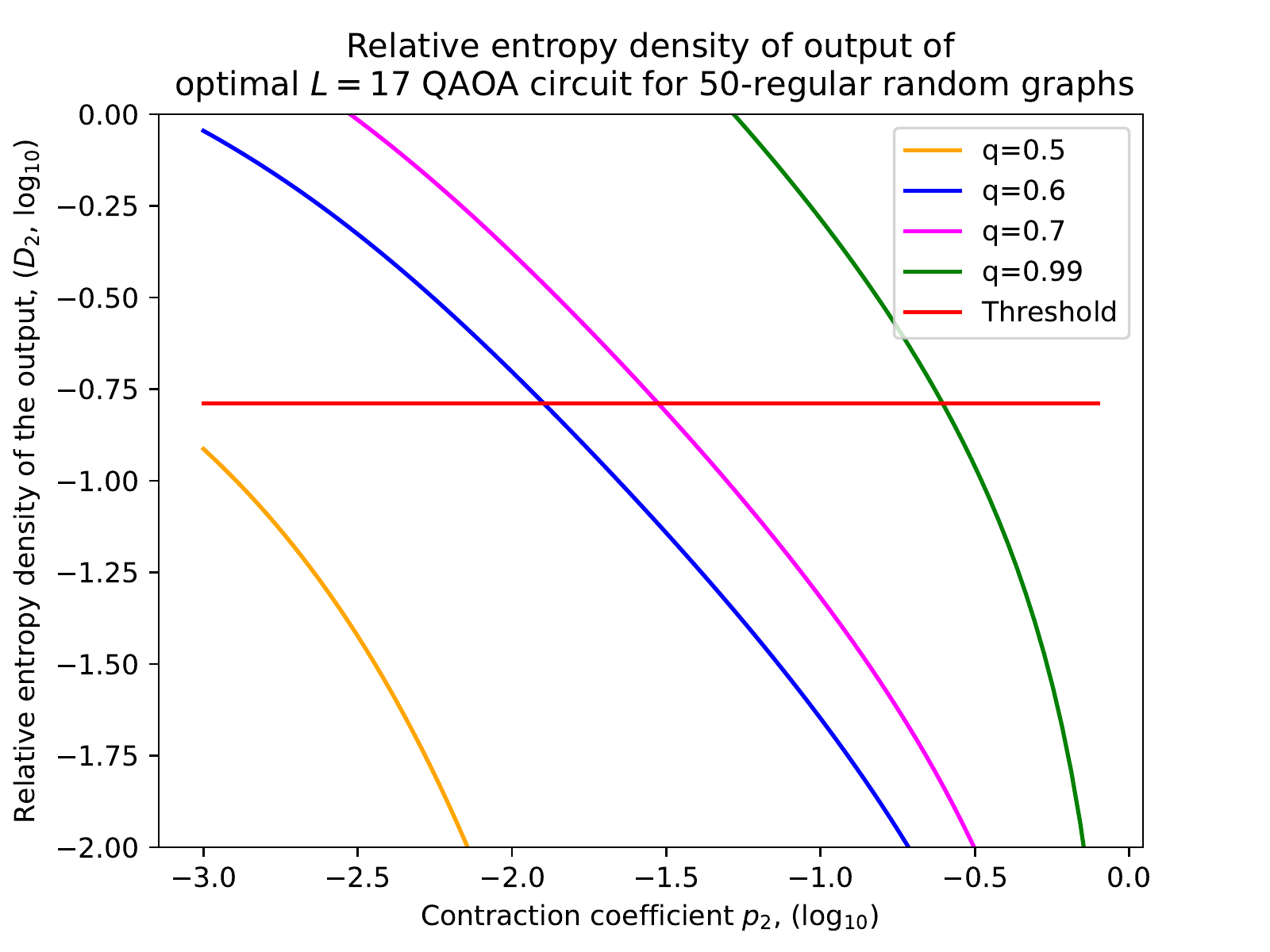}
    \caption{Relative entropy density of the output of a QAOA circuit of $P=17$ layers for various fixed-points $\tau_q^{\otimes n}$ as a function of the contraction coefficient and $D=50$. We used the optimal parameters found in~\cite{Basso2021a} for our circuit. The threshold we used is the one in~\eqref{equ:threshold_graph_main} and we used \autoref{cor:QAOA_rel_ent} to estimate the relative entropy decay. Although we see that our bounds have a worse performance as $q\to 1$, the amount of noise we can tolerate is still independent of the system's size.}
    \label{fig:plot_noisyqaoa}
\end{figure}
Although \autoref{fig:plot_noisyqaoa} seems to suggest that advantage is only lost at high noise levels as the fixed point becomes purer, recall that when implementing the QAOA circuit on the actual device, the circuit depth will be significantly larger than $17$. Indeed, in the plot, we took $D=50$, which means that a circuit of depth at least $50$ of two-qubit gates is required to implement each layer of $e^{i\gamma_i H_It}$. If we further incorporate the compilation of gates and the fact that NISQ devices are unlikely to have all-to-all connectivity, which imposes extra layers of SWAP gates, the depth required to implement each layer of QAOA with $D=50$ will conservatively be of order at least $10^2$. Thus, it is also reasonable to assume that the effective noise rate when implementing a layer of the QAOA circuit will be two orders of magnitude larger than the physical noise rate.

More generally, our bounds predict that quantum advantage will be lost whenever the QAOA parameters satisfy $\beta_k\to 0$ as $k\to\infty$. This is the case for the optimal parameters found in~\cite{Basso2021a}. This is because for such parameters the relative entropy between the output of the circuit and $\tau_q^{\otimes n}$ decays to $0$. This is illustrated more clearly in the continuous-time case of quantum annealing we discuss now.
\subsection{Noisy quantum annealing beyond unital noise}\label{subsec:example_annealing}
In this subsection, we will illustrate the bound in \autoref{prop:entropy_decay_annealer_main} for the case of noisy annealers with a linear schedule. That is, the function $f$ in the statement is just given by $f(t)=(1-t)$. Furthermore, we will assume that the time-independent Lindlbadian of spectral gap $1$ is driving the system to the product state $\tau_q^{\otimes q}$ with $\tau_q=q\ketbra{0}+(1-q)\ketbra{1}$ for $q<\tfrac{1}{2}$.
\begin{prop}\label{prop:noisy_annealer_concent_main}
For $0< q\leq \tfrac{1}{2}$ and $T>0$ let
\begin{align}
r_2=2\frac{1-q}{\log(q^{-1})}\,.
\end{align}
and $h(T)$ be
\begin{align}\label{equ:def_h}
&h(T)=e^{-r_2T}\log\left(\frac{ 1+2 \left(q-q^2\right)^{\frac{1}{2}}}{4\left(q-q^2\right)}\right)+\nonumber\\&\frac{(2q-1)(1-e^{-r_2T}rT-e^{-r_2T})}{(q(1-q))^{\tfrac{1}{2}}r_2^2T}.
\end{align}
Furthermore, let $\mathcal{T}_t$ be defined as in \autoref{prop:entropy_decay_annealer_main} and $f(t)=(1-t)$. Then for the initial state $\ket{+}^{\otimes n}$ and $
\rho_T=\mathcal{T}_T(\ketbra{+}^{\otimes n})$ we have:
\begin{align}
&\mathbb{P}_{\rho_T}(H_I\leq(\tr(\left[\tau_q^{\otimes n}H_I\right]-2^{-\frac{1}{2}}(h(T)+\epsilon)^{1/2}\|H_I\|_{\operatornamewithlimits{Lip}}n)\mathbb{I})\nonumber\\&\leq \operatorname{exp}\left(-\frac{\epsilon n}{2}\right).
\end{align}
\end{prop}
We refer to \autoref{sec:computations_relevant_Example} for a discussion of this result and \autoref{prop:concentration_noisy_annealer} in the same section for a proof. But the take-away message from \autoref{prop:noisy_annealer_concent_main} is that we can still derive concentration inequalities beyond unital noise. However, the bounds get looser as $q\to0$ (i.e. the fixed point becomes pure) and the decay of the relative entropy is polynomial instead of exponential.

We can reach similar conclusions for the purity of the output and, thus, for the probability that virtual cooling succeeds.
\begin{prop}\label{prop:purity_noisy_annealer_main}
For $0< q\leq \tfrac{1}{2}$ and $T>0$ let
\begin{align}
    r_2=2\frac{1-q}{\log(q^{-1})}\,.
\end{align}
and $h(T)$ be as in Eq.~\eqref{equ:def_h}.
Furthermore, let $\mathcal{T}_t$ be defined as in \autoref{prop:entropy_decay_annealer} and $f(t)=(1-t)$. For the initial state $\ket{+}^{\otimes n}$ let $T$ be large enough for $h(T)\leq 1-\log(2(1-q))-\epsilon$ to hold for some $\epsilon>0$. Then the probability that virtual cooling or distillation succeeds is at most $\operatorname{exp}(-\epsilon n)$.

\end{prop}
We refer to \autoref{sec:computations_relevant_Example} for a proof.

\section{Conclusion and Open Problems}
In this work we have used techniques of quantum optimal transport to derive various concentration inequalities for quantum circuits. In particular, we showed quadratic concentration for shallow circuits and Gaussian concentration for noisy circuits at large enough depth and Lipschitz observables.

By applying such inequalities to variational quantum algorithms such as QAOA or quantum annealing algorithms, we showed that for most instances, the probability that these algorithms outperform classical algorithms is exponentially small whenever the circuit has a nontrivial density of errors. Furthermore, we obtained self-contained and simplified proofs of previous results on the limitations of QAOA.

Our work demonstrates the relevance of quantum optimal transport methods to near-term quantum computing. Furthermore, it closes a few important gaps in previous results on limitations of variational quantum algorithms. 

An important problem that is left by our work is whether it is also possible to obtain Gaussian concentration inequalities for the outputs of shallow circuits. After the posting the first version of the present work, the authors of \cite{anshu2022concentration} found a different method based on polynomial approximations for showing that the output distributions in fact satisfy a stronger Gaussian concentration bound, hence answering this question.

\section{Acknowledgments}
GDP is a member of the ``Gruppo Nazionale per la Fisica Matematica (GNFM)'' of the ``Istituto Nazionale di Alta Matematica ``Francesco Severi'' (INdAM)''.
MM acknowledges support by the NSF under Grant No.CCF-1954960 and by IARPA and DARPA via the U.S. Army Research Office contract W911NF-17-C-0050.
DSF acknowledges financial support from the VILLUM FONDEN via the QMATH Centre of Excellence (Grant no. 10059)  and the QuantERA ERA-NET Cofund in Quantum Technologies implemented within the European Unionâ€™s Horizon 2020 Program (QuantAlgo project) via the Innovation Fund Denmark. 
CR acknowledges financial support
from a Junior Researcher START Fellowship from the DFG cluster of excellence 2111 (Munich Center for Quantum Science
and Technology), from the ANR project QTraj (ANR-20-CE40-0024-01) of the French National Research Agency (ANR), as well as from the Humboldt Foundation.

\bibliographystyle{abbrv}
\bibliography{references}

\newpage 
\onecolumngrid
\appendix

\section{Notations}\label{notationsappendix}

We consider a set $V$ corresponding to a system of $|V|=n$ qudits, and denote by $\cH_V=\bigotimes_{v\in V}\mathbb{C}^d$ the Hilbert space of $n$-qudits and by $\cB(\cH_V)$ the algebra of linear operators on $\cH_V$. 
$\mathcal{O}_V$ corresponds to the self-adjoint linear operators on $\cH_V$, whereas $\mathcal{O}^T_V\subset \mathcal{O}_V$ is the subspace of traceless self-adjoint linear operators. $\mathcal{O}_V^+$ denotes the subset of positive semidefinite linear operators on $\cH_V$ and $\mathcal{S}_V\subset \mathcal{O}_V^+$ denotes the set of quantum states. 
Similarly, we denote by $\mathcal{P}_V$ the set of probability measures on $[d]^V$. Given an operator $X\in\cB(\cH_V)$, we denote by $X^\dagger$ its adjoint with respect to the inner product of $\cH_V$. 
Similarly, the adjoint of a linear map $\cN:\cB(\cH_V)\to\cB(\cH_V)$ with respect to the trace inner product is denoted by $\cN^\dagger$. 
For any subset $A\subseteq V$, we use the standard notations $\mathcal{O}_A, \mathcal{S}_A\ldots$ for the corresponding objects defined on subsystem $A$. 
Given a state $\rho\in\mathcal{S}_V$, we denote by $\rho_A$ its marginal on subsystem $A$. For any $X\in\mathcal{O}_V$, we denote by $\|X\|_p$ its Schatten $p$ norm. For any region $A\subset V$, the identity on $\mathcal{O}_{A}$ is denoted by $\mathbb{I}_A$, or more simply $\mathbb{I}$. 
Given an observable $O$, we define $\langle O\rangle_\sigma=\tr\left[ \sigma O\right]$. Moreover, given a number $a\in\mathbb{R}$, we denote $\{O\ge a\}$ to be the projector onto the subspace spanned by the eigenvectors of $O$ corresponding to eigenvalues greater than or equal to $a$. 
We denote the probability of measuring an eigenvalue of $O$ greater than $a\in \mathbb{R}$ in the state $\sigma$ as $\mathbb{P}_\sigma(O\ge a):=\tr\big[\sigma\{O\ge a\}\big]$. Given two probability measures $\mu,\nu$ over a common measurable space, $\mu<\!<\nu$ means that $\mu$ is absolutely continuous with respect to $\nu$. We will make use of the sandwiched R\'enyi divergences~\cite{MuellerLennert2013,Wilde2014} of order $\alpha\in(1,+\infty)$. For two states $\rho,\sigma$ such that the support of $\rho$ is included in the support of $\sigma$ they are defined as
\begin{align*}
    D_\alpha(\rho\|\sigma)=\frac{1}{\alpha-1}\log \tr{\left[\left(\sigma^{\frac{1-\alpha}{2\alpha}}\rho\sigma^{\frac{1-\alpha}{2\alpha}}\right)^\alpha\right]}.
\end{align*}
We will also consider the relative entropy we obtain by taking the limit $\alpha\to\infty$,
\begin{align*}
    D_{\infty}(\rho\|\sigma)=\log(\|\sigma^{-\frac{1}{2}}\rho\sigma^{-\frac{1}{2}}\|_\infty)
\end{align*}
and the usual Umegaki relative entropy between two quantum states $\rho,\sigma$, defined as 
\begin{align*}
    D(\rho\|\sigma):=\tr[\rho\,(\log\rho-\log\sigma)]\,,
\end{align*}
which corresponds to the limit $\alpha\to1$. In case the support of $\rho$ is not contained in that of $\sigma$, all the divergences above are defined to be $+\infty$.

\section{R\'enyi divergences and concentration inequalities}

In this section, we will show how to use R\'enyi divergences to transfer results about concentration from one state to another.
These divergences can be used to transfer concentration inequalities between states as follows:
\begin{lem}[Transferring concentration inequalities]\label{lem:transfer_concent}
Let $\rho$ and $\sigma>0$ be two quantum states on $\mathcal{H}_V$. Then for any POVM element $0\le E\le \mathbb{I}$ and $\alpha>1$ we have:
\begin{align}\label{equ:prob_rho_renyi}
    \tr\left[E\rho\right]\leq \operatorname{exp}\left[\frac{\alpha-1}{\alpha}\big(D_\alpha(\rho\|\sigma)+\log(\tr\left[E\sigma\right])\big)\right].
\end{align}
In particular, if $\sigma$ satisfies the Gaussian concentration inequality
\begin{align*}
        \mathbb{P}_\sigma\big(|O-\langle O\rangle_\sigma|\ge a |V|  \big)\leq K\,\exp{\left(-\frac{ca^2|V|}{\|\sigma^{-\frac{1}{2}}O\sigma^{\frac{1}{2}}\|_{L}^2}\right)}
\end{align*}
for some constants $c,K> 0$, then for any $\alpha>1$:
\begin{align}\label{equ:concentration_from_noise}
    \mathbb{P}_\rho\big(|O-\langle O\rangle_\sigma|\ge a |V|\big)\leq \operatorname{exp}\left[\frac{\alpha-1}{\alpha}\left(D_\alpha(\rho\|\sigma)-\frac{ca^2|V|}{\|\sigma^{-\frac{1}{2}}O\sigma^{\frac{1}{2}}\|_{L}^2}+\log(K)\right)\right]\,.
\end{align}
\end{lem}
\begin{proof}
We have:
\begin{align*}
  \tr\left[E\rho\right]=\tr\left[\sigma^{-\frac{1-\alpha}{2\alpha}}E\sigma^{-\frac{1-\alpha}{2\alpha}}\sigma^{\frac{1-\alpha}{2\alpha}}\rho\sigma^{\frac{1-\alpha}{2\alpha}}\right]\leq  \tr\left[\left(\sigma^{-\frac{1-\alpha}{2\alpha}}E\sigma^{-\frac{1-\alpha}{2\alpha}}\right)^{\alpha'}\right]^{\frac{1}{\alpha'}}\tr\left[\left(\sigma^{\frac{1-\alpha}{2\alpha}}\rho\sigma^{\frac{1-\alpha}{2\alpha}}\right)^\alpha\right]^{\frac{1}{\alpha}}
\end{align*}
by an application of H\"older's inequality and $\alpha'$ here being the H\"older conjugate of $\alpha$. Next, by the Araki-Lieb-Thirring inequality:
\begin{align}
    \tr\left[\left(\sigma^{-\frac{1-\alpha}{2\alpha}}E\sigma^{-\frac{1-\alpha}{2\alpha}}\right)^{\alpha'}\right]\leq \tr\left[\sigma^{-\frac{(1-\alpha)\alpha'}{2\alpha}}E^{\alpha'}\sigma^{-\frac{(1-\alpha)\alpha'}{2\alpha}}\right]\leq \tr\left[\sigma^{-\frac{(1-\alpha)\alpha'}{2\alpha}}E\sigma^{-\frac{(1-\alpha)\alpha'}{2\alpha}}\right],
\end{align}
where in the last inequality we used the fact that $\alpha'>1$ and $E\leq \mathbb{I}$. Furthermore, as $\alpha'$ is the H\"older conjugate of $\alpha$, we have that $\tfrac{1}{\alpha'}=\frac{\alpha-1}{\alpha}$ and then:
\begin{align*}
    \tr\left[\sigma^{-\frac{(1-\alpha)\alpha'}{2\alpha}}E\sigma^{-\frac{(1-\alpha)\alpha'}{2\alpha}}\right]=\tr\left[\sigma E\right].
\end{align*}
The claim in Eq.~\eqref{equ:prob_rho_renyi} then follows from a simple manipulation and by noting that 
\begin{align*}
    \tr\left[\left(\sigma^{\frac{1-\alpha}{2\alpha}}\rho\sigma^{\frac{1-\alpha}{2\alpha}}\right)^\alpha\right]^{\frac{1}{\alpha}}=\operatorname{exp}\left(\frac{\alpha-1}{\alpha} D_\alpha(\rho\|\sigma)\right).
\end{align*}
Eq.~\eqref{equ:concentration_from_noise} also immediately follows from plugging in the Gaussian concentration bound.
\end{proof}

\section{Entropic convergence results}\label{app:entropic_convergence}
In this section we will collect some results that allow us to estimate the sandwiched R\'enyi divergence between the output of a noisy quantum circuit or annealer and the fixed point of the noise affecting the device. 
In essence, these results are a generalization of the results of~\cite[Lemma 1 and Theorem 1]{Franca2021a}. In that work, the authors show precisely the same bounds as here, but only for the Umegaki relative entropy. 
However, their proofs can immediately be adapted to our setting with R\'enyi divergences. Thus, we will restrict ourselves to showing how to obtain a convergence result for discrete time circuits and do not describe the same proof for continuous-time in full detail. 
\begin{lem}[Lemma 1 of~\cite{Franca2021a}]\label{lem:dataprocessedtriangle}
    Let $\mathcal{N}:\cB(\cH_V)\to\cB(\cH_V)$ be a quantum channel with unique fixed point $\sigma>0$ that satisfies a strong data-processing inequality with constant $p_\alpha>0$ for some $\alpha>1$. That is,
    \begin{align}\label{equ:contraction}
    D_\alpha(\mathcal{N}(\rho)\|\sigma)\leq (1-p_\alpha)D_\alpha(\rho\|\sigma)
    \end{align}
    for all states $\rho$. Then for any other quantum channels $\Phi_1,\ldots,\Phi_m:\cB(\cH_V)\to\cB(\cH_V)$ we have:
    \begin{align}\label{equ:dataproccphi}
    &D_\alpha\Big(\prod_{t=1}^m (
    \Phi_t\circ  \mathcal{N})(\rho)\Big\|\sigma\Big)\leq (1-p_\alpha)^m D_\alpha(\rho\|\sigma)+\sum_{t=1}^{m} (1-p_\alpha)^{m-t} D_{\infty}( \Phi_{t}(\sigma)\|\sigma)\,.
    \end{align}
    \end{lem}
    \begin{proof}
    For $m=1$, this follows from the data-processed triangle inequality of~\cite[Theorem 3.1]{christandl_relative_2017}.
    In their notations, it states that for any quantum channel $P$, states $\rho,\sigma,\sigma'$ and $\alpha\geq1$ we have:
    \begin{align*}
        D_\alpha(P(\rho)\|\sigma)\leq D(\rho\|\sigma')+D_{\infty}(P(\sigma')\|\sigma).
    \end{align*}
    Setting $P=\Phi_1$ and $\sigma'=\sigma$ in their notation it implies that:
    \begin{align}\label{equ:resultn=1}
    D_\alpha((\Phi_1\circ  \mathcal{N})(\rho)\|\sigma)\leq D_\alpha(\mathcal{N}(\rho)\|\sigma)+ D_{\infty}(\Phi_1(\sigma)\|\sigma)\leq (1-p_\alpha)D_\alpha(\rho\|\sigma)+ D_{\infty}(\Phi_1(\sigma)\|\sigma).
    \end{align}
    Let us now assume the  claim to be true for some $m=k$. Then for $m=k+1$ we have:
    \begin{align}\label{equ:inductionhypok}
    D_\alpha\Big(\prod_{t=1}^{k+1} (
    \Phi_t\circ  \mathcal{N})(\rho)\Big\|\sigma\Big)\leq (1-p_\alpha)^k D_\alpha( (
    \Phi_{k+1}\circ  \mathcal{N})(\rho)\|\sigma)+\sum_{t=1}^{k} (1-p_\alpha)^{k-t} D_{\infty}( \Phi_{t}(\sigma)\|\sigma)
    \end{align}
    by our induction hypothesis. 
    Applying Eq.~\eqref{equ:resultn=1} to the first term in Eq.~\eqref{equ:inductionhypok}, the strong data-processing inequality, we obtain the claim.
    \end{proof}
Note that \autoref{lem:dataprocessedtriangle} implies that the R\'enyi divergence will converge to $0$ whenever $\Phi_{t}(\sigma)\simeq \sigma$ as $t\to\infty$. 
This is always the case for unitary circuits under unital noise, as the fixed point is the maximally mixed state and is invariant under unitaries, but is also expected to hold for QAOA circuits. See \autoref{sec:examples} for examples of such circuits.

We can also show similar statements for continuous-time evolutions under noise to also study quantum simulators or annealers:
\begin{lem}[Theorem 1 of~\cite{Franca2021a}]\label{thm:continuumlimit}
    Let $\cL:\cB(\cH_V)\to\cB(\cH_V)$ be a Lindbladian with fixed point $\sigma$. Suppose that for some $\alpha>1$ we have for all $t>0$ and initial states that there is a $r_\alpha>0$ such that:
    \begin{align}\label{equ:renyi_entropy_production}
        D_\alpha(e^{t\cL}(\rho)\|\sigma)\leq e^{-r_\alpha t}D_\alpha(\rho\|\sigma).
    \end{align}
    
    Moreover, let $\cH_t:\cB(\cH_V)\to\cB(\cH_V)$ be given by $\cH_t(X)=i[X,H_t]$ for some time-dependent Hamiltonian $H_t$. Moreover, let $\mathcal{T}_t$ be the evolution of the system under the Lindbladian $\cS_{t}=\cL+\cH_t$ from time $0$ to $t$. Then for all states $\rho$ and times $t>0$:
    \begin{align}\label{equ:ourintegral}
    D_\alpha(\mathcal{T}_t(\rho)\|\sigma)\leq e^{-r_\alpha t}D_\alpha(\rho\|\sigma)+\int\limits_{0}^t  e^{-r_\alpha(t-\tau)}\|\sigma^{-\frac{1}{2}}[\sigma,H_\tau]\sigma^{-\frac{1}{2}}\|_\infty\,d\tau\,.
    \end{align}
    \end{lem}

Thus, armed with contraction inequalities like those in Eq.~\eqref{equ:renyi_entropy_production} or Eq.~\eqref{equ:contraction} it is straightforward to obtain estimates on R\'enyi entropies. For completeness, we will collect some known results and techniques to obtain such contraction inequalities in the next Section.

\subsection{Contraction results for sandwiched R\'enyi divergences}
Let us now collect some known results to obtain inequalities like Eq.~\eqref{equ:renyi_entropy_production} or Eq.~\eqref{equ:contraction}. 
We will focus on the case where the noise has a product form, i.e. $\mathcal{N}=\bigotimes_{i=1}^{n}\mathcal{N}_i$, where $\mathcal{N}_i$ acts only on qubit $i$. Although it is straightforward to generalize the results to the case in which there is a different channel acting on each qubit, we will make the simplifying assumption that all local channels are the same. 
Furthermore, we will focus on inequalities that tensorize. 
This means that $q_\alpha$ will not scale with the size of the system $n$. To the best of our knowledge, strong data processing inequalities are not available for R\'enyi entropies beyond product channels.

Let us start with the continuous-time setting, as more is known there.
For continuous-time, the contraction of R\'enyi entropies was systematically studied in~\cite{MllerHermes2018}. 
In particular, in~\cite[Theorem 4.3]{MllerHermes2018} the authors relate bounds on the optimal decay rate $r_\alpha$ to so-called logarithmic Sobolev inequalities~\cite{Olkiewicz1999,Carbone2015,Kastoryano2013,Temme2014}. 
It is beyond the scope of this article to review logarithmic Sobolev inequalities and we focus instead on the contraction rate these tools give to the problem at hand.

If we have a Lindbladian of the form
\begin{align*}
\cL^{(n)}=\cL\otimes \textrm{id}_{n-1}+\textrm{id}_1\otimes \cL\otimes\textrm{id}_{n-2}+\ldots+\textrm{id}_{n-1}\otimes \cL
\end{align*}
with unique fixed point $\otimes_{i=1}^n\sigma$, then 
\begin{align}\label{equ:contraction_rate_d2}
D_2(e^{t\cL^{(n)}}(\rho)\|\otimes_{i=1}^n\sigma)\leq e^{-r_2t}D_2(\rho\|\otimes_{i=1}^n\sigma)
\end{align}
holds with
\begin{align}\label{equ:beta2forlocal}
r_2=2\lambda(\cL)\,\frac{1-\frac{1}{\|\sigma^{-1}\|}}{\log(\|\sigma^{-1}\|)}\,,
\end{align}
where $\lambda(\cL)$ is the spectral gap of the local Linbladian $\cL$. For instance, for generalized depolarizing noise we have $\lambda(\cL)=1$. The take-home message of Eq.~\eqref{equ:beta2forlocal} is that as long as $\|\sigma^{-1}\|=\cO(1)$, the rate with which the sandwiched R\'enyi-2 divergence contracts is constant as well. It is also possible to use similar tools to derive the contraction for other values of $\alpha>1$ and we refer to~\cite{MllerHermes2018,Hirche2020} for a more detailed discussion. 
However, to the best of our knowledge, all known results exhibit a similar scaling as that in Eq.~\eqref{equ:beta2forlocal} and we do not discuss this further.

In discrete time, the best results available are to the best of our knowledge those of~\cite[Corollary 5.5, 5.6]{Hirche2020}. To parse their results we first need to introduce some notation. 
For a given $\sigma$ we will denote by $\Gamma^{\alpha}_\sigma:\cB(\cH)\to\cB(\cH)$ the map $X\mapsto \sigma^{\frac{\alpha}{2}}X\sigma^{\frac{\alpha}{2}}$ and by $\cD_{p,\sigma}$ the generalized depolarizing channel converging to the state $\sigma$ (i.e. $\rho\mapsto (1-p)\rho+p\sigma$). 
It follows from~\cite[Corollary 5.6]{Hirche2020} that if for a quantum channel $\mathcal{N}_i$ with fixed point $\sigma$ we have
\begin{align}\label{equ:contraction_norm_channel}
\|\Gamma^{-\frac{1}{2}}\circ \mathcal{N}_i\circ \cD_{p,\sigma}^{-1}\circ\Gamma^{\frac{1}{2}}\|_{2\to 2}\le 1\,,
\end{align}
then for any state $\rho$ on $n$ qudits:
\begin{align}\label{equ:contraction_discrete}
D_2(\otimes_{i=1}^n \mathcal{N}_i(\rho)\|\otimes_{i=1}^n \sigma)\leq (1-p)^{\frac{\|\sigma^{-1}\|_\infty-1}{\|\sigma^{-1}\|_\infty\log(\|\sigma^{-1}\|_\infty)}}D_2(\rho\|\otimes_{i=1}^n \sigma)\,.
\end{align}
The expressions in Eq.~\eqref{equ:contraction_norm_channel} and Eq.~\eqref{equ:contraction_discrete} may seem daunting at first, so let us digest them a bit further and summarize their message. 
First, note that Eq.~\eqref{equ:contraction_norm_channel} only involves one copy of the quantum channel, whereas the expression in Eq.~\eqref{equ:contraction_discrete} involves arbitrarily many. 
Thus, this is an example of an inequality that tensorizes. Furthermore, note that Eq.~\eqref{equ:contraction_norm_channel} can be verified efficiently. 
This is because it just corresponds to checking whether the operator norm of a linear operator is smaller than or equal to one or not, which can be computed in polynomial time. 
Thus, by performing a binary search on the values of $p$ for which the inequality holds, we can approximate the largest $p$ for which it holds. 
Then Eq.~\eqref{equ:contraction_discrete} tells us that once we establish such an inequality, the R\'enyi-2 divergence will contract by a rate that is independent of the system size. The take-home message of Eq.~\eqref{equ:contraction_discrete} is essentially the same as that of Eq.~\eqref{equ:contraction_rate_d2}. 
As long as $\|\sigma^{-1}\|_\infty=\cO(1)$, the R\'enyi-2 divergence will contract with a constant rate. This corresponds to the setting in which each local fixed point does not have a purity scaling with system size.

\subsection{Specializing \autoref{lem:dataprocessedtriangle} to QAOA and quantum annealing}\label{sec:qaoa_anneal}

In the main text we only considered quantum circuits that are affected by unital noise. The reason for that is that then one can use \autoref{lem:dataprocessedtriangle} to obtain the exponential decay of the relative entropy to the maximally mixed state independently of the circuit that is being implemented.

However, it is still possible to obtain closed formulas for the relative entropy decay for QAOA-like circuits, as we will show now. We are still going to depart from the assumption that the noise affecting the device has a product state $\sigma_q=\otimes_{i=1}^n \tau_q$ as its fixed point, with
\begin{align}\label{equ:generalized_fixed}
\tau_q=q\ketbra{0}+(1-q)\ketbra{1}
\end{align}
for some $q\in[0,1]$.

Recall that for $H_I$ the Ising Hamiltonian whose energy we wish to minimize, $H_X=-\sum_iX_i$ and $\gamma,\beta\in\R^P$, the QAOA unitary is given by:
\begin{equation}
    V_{\gamma, \beta}=\prod\limits_{k=1}^P e^{i\beta_k H_X}e^{i\gamma_k H_I}.
\end{equation}
In order to obtain an estimate of the relative entropy decay under a noisy version of this circuit, we need to analyse the expressions:
\begin{align*}
D_{\infty}(e^{i\gamma_k H_I}\sigma_q e^{-i\gamma_k H_I}\|\sigma_q),\quad D_{\infty}(e^{i\beta_k H_X}\sigma_q e^{-i\beta_k H_X})\|\sigma_q).
\end{align*}
We then have:
\begin{lem}\label{lem:entropy_paulix}
  Let $\beta,\gamma\in\R^P$ be given and for $q\in[0,1]$ $\sigma_q$ as in Eq.~\eqref{equ:generalized_fixed}. Moreover, for $\beta_k,q$ define $z(\beta_k,q)$ as 
  \begin{align*}
 z(\beta_k,q)=2\cos(2\beta_k)+\frac{\sin^2(\beta_k)}{q(1-q)}.
  \end{align*}
  Then:
\begin{align}\label{equ:relative_entropy_QAOA_layer}
    D_{\infty}(e^{i\beta_k H_X}\sigma_q e^{-i\beta_k H_X}\|\sigma_q)=
n\log\left(\frac{z(\beta_k,q)+\sqrt{z(\beta_k,q)^2-4}}{2}\right).
\end{align}
    
\end{lem}
\begin{proof}
As both $e^{i\beta_k H_X}$ and $\sigma$ are of tensor product form, we obtain by the additivity of the max relative entropy that
\begin{align*}
D_{\infty}(e^{i\beta_k H_X}\sigma e^{-i\beta_k H_X})\|\sigma)=nD_{\infty}(e^{-i\beta_kX}\tau_q e^{i\beta_k H_X}\|\tau_q).
\end{align*}
A simple yet tedious computation shows that:
\begin{align}
\|\tau_q^{-\frac{1}{2}}e^{-i\beta_kX}\tau_q e^{i\beta_k H_X}\tau_q^{-\frac{1}{2}}\|=\frac{z(\beta_k,q)+\sqrt{z(\beta_k,q)^2-4}}{2}.
\end{align}
Taking the logarithm yields the claim.
\end{proof}
Before we state the entropy decay we obtain for QAOA circuits, let us briefly comment on the scaling of Eq.~\eqref{equ:relative_entropy_QAOA_layer}. First, note that either in the limit $q\to\tfrac{1}{2}$ or $\beta_k\to0$ we have that the r.h.s. of Eq.~\eqref{equ:relative_entropy_QAOA_layer}. The first case corresponds to the fixed point being the maximally mixed state, but the second corresponds to mixer unitaries for which the total time evolution is small.

On the other hand, if we let $q\to0$ or $q\to1$, then we see that the r.h.s. of Eq.~\eqref{equ:relative_entropy_QAOA_layer} goes to infinity. We then have:
\begin{cor}[Relative entropy decay for QAOA]\label{cor:QAOA_rel_ent}
    Let $\beta,\gamma\in\R^P$ be given and $\tau_q$ and $z$ defined as before. Moreover, let $\mathcal{N}$ be such that
    \begin{align}
        D_\alpha(\mathcal{N}(\rho)\|\sigma_q)\leq (1-p_\alpha)D_\alpha(\rho\|\sigma_q).
    \end{align}
Then for any initial state $\rho$ we have:
\begin{align}\label{equ:dataproccphiX}
    &D_\alpha\Big(\prod_{k=1}^P (
    e^{i\beta_kH_I}\circ  \mathcal{N}\circ e^{i\beta_kH_X}\circ  \mathcal{N})(\rho)\Big\|\sigma_q\Big)\leq \nonumber\\&\qquad\qquad\qquad \qquad  (1-p_\alpha)^{2P} D_\alpha(\rho\|\sigma_q)+\sum_{k=1}^{P} (1-p_\alpha)^{2(P-k)}n\log\left(\frac{z(\beta_k,q)+\sqrt{z(\beta_k,q)^2-4}}{2}\right).
    \end{align}
Furthermore, for the case of $\rho=\ketbra{+}^{\otimes n}$ and $\alpha=2$, we have     
\begin{align}\label{equ:relative_entropy_plus}
    D_2(\ketbra{+}^{\otimes n}\|\sigma_q)=n\log\left(\frac{q^{-1}+(1-q)^{-1}+2(q(1-q))^{-\frac{1}{2}}}{4}\right).
\end{align}
\end{cor}
\begin{proof}
    The first step is to observe that $D_{\infty}(e^{i\gamma_k H_I}\sigma_q e^{-i\gamma_k H_I})\|\sigma)=0$. This follows from the fact that $e^{i\beta_k H_I}$ is a diagonal unitary and, thus, commutes with $\sigma_q$. The claim then follows from combining \autoref{lem:dataprocessedtriangle} and the result of \autoref{lem:entropy_paulix}. To obtain the expression in Eq.~\eqref{equ:relative_entropy_plus}, note that $D_2$ tensorizes and the two underlying states are product. Thus, we only need to compute $D_2(\ketbra{+}\|\tau_q)$, a simple computation.
\end{proof}

From our previous discussion, it is straightforward to identify the conditions under which Eq.~\eqref{equ:dataproccphiX} converges to $0$ as $P\to\infty$. First, the case $q=\tfrac{1}{2}$, which corresponds to unital noise and we already covered at length in the main text. Second, whenever we have that $\beta_k\to0$ as $k\to\infty$. This is because the relative entropy terms in Eq.~\eqref{equ:dataproccphiX} at depth $k$ are suppressed by $(1-p_\alpha)^{2(P-k)}$. Thus, only at depths $k\simeq P$ the relative entropy is not suppressed. 

Interestingly, parameters $\beta,\gamma$ for which QAOA is expected to perform well fulfill this condition~\cite{Farhi2019,Basso2021}. To see this, it is fruitful to interpret QAOA as a trotterized version of quantum annealing, where we start with the Hamiltonian $H_X$ and adiabatically modify it to $H_I$. It is then clear that at late times of the computation, the Hamiltonian will approximate $H_I$ and the fixed point of the noise will be approximately preserved by the unitary evolution.

We can make this precise by deriving the analogous version of~\autoref{cor:QAOA_rel_ent} for quantum annealing:

\begin{prop}\label{prop:entropy_decay_annealer}
        Let $\cL:\cB(\cH_V)\to\cB(\cH_V)$ be a Lindbladian with fixed point $\sigma_q$ defined as before with $q\geq\tfrac{1}{2}$. Suppose that for some $\alpha>1$ we have for all $t>0$ and initial states that there is a $r_\alpha>0$ such that:
        \begin{align}\label{equ:renyi_entropy_production1}
            D_\alpha(e^{t\cL}(\rho)\|\sigma)\leq e^{-r_\alpha t}D_\alpha(\rho\|\sigma).
        \end{align}
        
        Moreover, for functions $f,g:[0,1]\to\R$ and $T>0$ let $\cH_t:\cB(\cH_V)\to\cB(\cH_V)$  be given by $\cH_t(X)=i[X,f(t/t)H_X+g(t/T)H_I]$. Let $\mathcal{T}_t$ be the evolution of the system under the Lindbladian $\cS_{t}=\cL+\cH_t$ from time $0$ to $t\leq T$. Then for all states $\rho$:
        \begin{align}\label{equ:bound_entropy_annealer_beyond}
        D_\alpha(\mathcal{T}_T(\rho)\|\sigma)\leq e^{-r_\alpha T}D_\alpha(\rho\|\sigma)+2ne^{-r_\alpha T}\left(\sqrt{\frac{p}{1-p}}-\sqrt{\frac{1-q}{q}}\right)\int\limits_{0}^T  e^{r_\alpha t}|f(t/T)|\,dt\,.
        \end{align}    
\end{prop}
\begin{proof}
From~\autoref{thm:continuumlimit} we see that all we need to obtain the claim is to estimate 
\begin{align*}
    \int\limits_{0}^T  e^{-r_\alpha(T-t)}\|\sigma_q^{-\frac{1}{2}}[\sigma_q,H_t]\sigma_q^{-\frac{1}{2}}\|_\infty\,dt
\end{align*}
As before, because $[H_I,\sigma_q]=0$, this simplifies to 
\begin{align*}
    \int\limits_{0}^T  e^{-r_\alpha(T-t)}\|\sigma_q^{-\frac{1}{2}}[\sigma_q,H_\tau]\sigma_q^{-\frac{1}{2}}\|_\infty\,dt&=    \int\limits_{0}^T  e^{-r_\alpha(T-t)}|f(t/T)|\|\sigma_q^{-\frac{1}{2}}[\sigma_q,H_X]\sigma_q^{-\frac{1}{2}}\|_\infty\,dt\\
    &\leq ne^{-r_\alpha T}\int\limits_{0}^T  e^{-r_\alpha t}|f(t/T)|\|\tau_q^{-\frac{1}{2}}[\tau_q,X]\tau_q^{-\frac{1}{2}}\|_\infty\,dt\,,
\end{align*}
where in the last step we applied a triangle inequality using $H_X=-\sum_iX_i$ and the fact that $\sigma_q=\otimes_{i=1}^n\tau_q$. 
The claim follows after noting that:
\begin{align*}
    \|\tau_q^{-\frac{1}{2}}[\tau_q,X]\tau_q^{-\frac{1}{2}}\|_\infty=\left(\sqrt{\frac{p}{1-p}}-\sqrt{\frac{1-p}{p}}\right).
\end{align*}
\end{proof}
As adiabatic theorems require that $f(1)=0$ to make sure that the we observe a good overlap with the ground state~\cite{Jansen_2007}, it follows that the R\'enyi entropy will typically decay to $0$ even under nonunital noise for quantum annealers. However, note once again that our bounds perform poorly whenever the fixed point is close to pure and whenever the function $f$ does not decay fast enough to $0$ around $1$.

\subsection{QAOA and quantum annealing on random regular graphs of high girth}\label{sec:regular_graphs}
In the previous section we established estimates on the relative entropy decay of QAOA circuits (\autoref{cor:QAOA_rel_ent}) and quantum annealers (\autoref{prop:entropy_decay_annealer}) under non-unital noise. Such estimates can then be combined with \autoref{thm:concentration_renyi} to obtain concentration inequalities for the outputs of these circuits.
One important caveat is that \autoref{cor:QAOA_rel_ent} and  \autoref{prop:entropy_decay_annealer} depend on the actual circuit being implemented. Thus, we cannot give universal bounds on the performance of such circuits that only depend on the depth and the noise level as it was the case for unital noise.

However, \autoref{cor:QAOA_rel_ent} can still be readily applied for a given choice of QAOA parameters and we will exemplify the performance of the bounds on QAOA on the Max-Cut of random $D$-regular graphs under noise. The motivation to study this particular class of instances is many. First, the asymptotic value of both the ground state energy and that of the standard SDP relaxation are known. It is known~\cite{Parisi1979,Dembo2015} that for the Ising model on a random $D$-regular graph on $n$ nodes the ground-state energy density scales like:
\begin{align}
-\Pi_\ast\sqrt{D}+o(\sqrt{D}),
\end{align}
with $\Pi_\ast=0.763166\ldots$ the Parisi constant. The value that assumption-free efficient classical algorithms~\cite{Thompson2021} achieve is given by $-\frac{2}{\pi}\sqrt{D}$ with $2/\pi\simeq 0.6366$. The fact that these values are known makes it straightforward to analyse at which energies the output of a noisy quantum algorithm will be outperformed by efficient classical algorithms.

Furthermore, there is a natural choice for the value of the QAOA parameters to pick for the circuit. Indeed, in~\cite{Basso2021a} the authors computed the optimal parameters for QAOA on such graphs for depths up to $P=17$. Note, however, that these are the optimal values as the system's size goes to infinity and in the absence of noise. Nevertheless, they provide a good testing ground for our bounds.

To start our analysis, note that if we define the one qubit state $\tau_q=q\ketbra{0}+(1-q)\ketbra{1}$ as before and let $H_{I,D}$ be the Ising Hamiltonian on a $D-$regular graph, then we have:
\begin{align}\label{equ:expectation_product}
\tr\left(H_{I,D}\tau_q^{\otimes n}\right)=(1-2q)^2\frac{nD}{2}.
\end{align}
To see this, note that the expectation value of each $Z_iZ_j$ term will be $(1-2q)^2$ and the graph is assumed to be $D$-regular. We then obtain Eq.~\eqref{equ:expectation_product} by noting that there are $\frac{nD}{2}$ edges in the graph. As the expected value of the energy achieved by classical algorithms is $-2n\sqrt{D}/\pi$, quantum advantage is lost if we deviate by less than $(1-2q)^2\frac{nD}{2}+2n\sqrt{D}/\pi$ from the mean under $\tau_q^{\otimes n}$.

We then have:
\begin{prop}\label{prop:threshold_relative_entropy}
Let $\rho$ be a quantum state on $n$ qubits and assume that for some $q\in(0,1)$, $\epsilon>0$ and $D>0$ we have 
\begin{align}\label{equ:bound_relative_entropy}
D_2(\rho\|\tau_q^{\otimes n})\leq \frac{((1-2q)^2\frac{D}{2}+2/\pi\sqrt{D})^2-\epsilon D^2}{2D^2}n.
\end{align}
Then the probability that the outcome of measuring $\rho$ in the computational basis provides a lower energy than efficient classical algorithms for Max-Cut on random $D$-regular high girth algorithms is at most $e^{-\tfrac{\epsilon}{2} n}$.
\end{prop}
\begin{proof}
Note that we have $\|H_{I,D}\|_{\operatorname{Lip}}=D$, as the graph is $D$-regular. 
By \autoref{thm:concentration_renyi} we have:
\begin{align}\label{equ:biased_concentration}
    \mathbb{P}_\rho\big(|H_{I,D}-(1-2q)^2\frac{D}{2}|\ge  a n  \big)\leq\operatorname{exp}\left(\frac{1}{2}\left(D_2(\rho\|\tau_q^{\otimes n})-\frac{a^2 n}{2D^2}\right)\right).
\end{align}
By our previous discussion, we know that we need to deviate from the mean at the state $\tau_q^{\otimes n}$ by at least $(1-2q)^2\frac{nD}{2}+2n\sqrt{D}/\pi$ before measuring the quantum state outperforms classical algorithms. Thus, we can pick $a=(1-2q)^2\frac{D}{2}+2\sqrt{D}/\pi$ in Eq.~\eqref{equ:concentration_transference} as our measure of when advantage is lost. It is then easy to see that for our bound on $D_2$ in Eq.~\eqref{equ:bound_relative_entropy}, we have that the r.h.s. of Eq.~\eqref{equ:biased_concentration} is $e^{-\tfrac{\epsilon}{2} n}$, which shows the claim.

\end{proof}

\autoref{prop:threshold_relative_entropy} allows us to conclude that if the output $\rho$ of a QAOA circuit satisfies 
\begin{align}\label{equ:threshold_graph}
    D_2(\rho\|\tau_q^{\otimes n})< \frac{((1-2q)^2\frac{D}{2}+\frac{2}{\pi}\sqrt{D})^2}{2D^2}n,
\end{align}
then quantum advantage is lost.
In~\cite[Table 4]{Basso2021} the authors give optimal parameters that in the noiseless case outperform known efficient classical algorithms. We can then insert these parameters into the bound obtained in \autoref{cor:QAOA_rel_ent} to estimate at which noise levels advantage is lost. 

It is important to stress once again that these parameters are only known to be optimal in the absence of noise and in the limit of nodes and degree going to infinity. However, we believe that they still provide a natural choice of parameters to analyse under noise. Importantly, note that Eq.~\eqref{equ:threshold_graph} once again only requires the relative entropy to contract by a constant factor before advantage is lost as long as $q\not=\frac{1}{2}$.

In \autoref{fig:plot_noisyqaoa} of the main text we plot the performance of QAOA with the parameters for $P=17$ as predicted by our bounds. In the absence of noise these QAOA circuits outperform efficient classical algorithms, but we show that this is not necessarily the case in the presence of noise.
Note that the values of $\gamma_i$ are irrelevant for the analysis. The values of $\beta_i$ we used are 
\begin{align}
    \beta=[&0.6375, 0.5197, 0.4697, 0.4499, 0.4255,0.4054, 0.3832, 0.3603, 0.3358, 0.3092,\nonumber\\&0.2807, 0.2501, 0.2171, 0.1816, 0.1426, 0.1001, 0.0536]\,.
\end{align}

\subsection{Computations required for \autoref{subsec:example_annealing}}\label{sec:computations_relevant_Example}
In this subsection we collect some auxiliary computations required to arrive at the conclusion of the example discussed in \autoref{subsec:example_annealing}.
Our goal is to evaluate the formula in Eq.~\eqref{equ:bound_entropy_annealer_beyond} for the case where the initial state is given by $\ket{+}=\tfrac{1}{\sqrt{2}}(\ket{0}+\ket{1})$ and the annealing schedule is linear, $f(t)=(1-t)$. Furthermore, for simplicity, we will assume that the local Lindbladians $\cL_i$ have as fixed point the state $\tau_q$ for $q\leq \tfrac{1}{2}$ and spectral gap $\lambda=1$. The results can then be easily rescaled to obtain the bounds for other values of the spectral gap.

The first observation we make is that under these assumptions Eq.~\eqref{equ:beta2forlocal} implies that:
\begin{align}\label{equ:beta2forlocal_recall}
    r_2\geq 2\frac{1-q}{\log(q^{-1})}\,.
\end{align}
Furthermore, a simple yet tedious calculation shows that:
\begin{align}
D_2(\ketbra{+}^{\otimes n}\|\tau_q^{\otimes n})=n\log\left(\frac{ 1+2 \left(q-q^2\right)^{\frac{1}{2}}}{4\left(q-q^2\right)}\right)
\end{align}
and the integral in Eq.~\eqref{equ:bound_entropy_annealer_beyond} evaluates to
\begin{align*}
\int\limits_{0}^T  e^{r_2 t}|f(t/T)|\,dt\,=\frac{e^{rT}-rT-1}{r^2T}.
\end{align*}

Putting all of these elements together we obtain the bound
\begin{align}\label{equ:bound_entropy_annealer_beyond2}
n^{-1}D_2(\mathcal{T}_T(\rho)\|\tau_q^{\otimes n})\leq e^{-r_2T}\log\left(\frac{ 1+2 \left(q-q^2\right)^{\frac{1}{2}}}{4\left(q-q^2\right)}\right)+\frac{(2q-1)(1-e^{-r_2T}rT-e^{-r_2T})}{(q(1-q))^{\tfrac{1}{2}}r_2^2T},
\end{align}
where $r_2$ is lower-bounded in Eq.~\eqref{equ:beta2forlocal_recall}. Furthermore, by combining ~\cite[Theorem 2]{de2020quantum} and~\cite[Theorem 7]{depalma2021quantum} we conclude for $\tau_q$ and $O$ satisfying $[O,\tau_q^{\otimes n}]=0$ we have that:
\begin{align}\label{equ:concentration_fixed_noise}
\mathbb{P}_{\tau_q^{\otimes n}}(|O-\tr\left[O\tau_q^{\otimes n}\right]\mathbb{I}|\geq r)\leq 2\textrm{exp}\left(-\frac{2r^2}{n\|O\|_{\operatornamewithlimits{Lip}}^2}\right).
\end{align}
The one-sided bound also holds without the prefactor $2$. Now that we have a contraction result for the R\'enyi divergence and a concentration inequality for the fixed point of the noise, it is straightforward to also obtain concentration bounds for the output of the noisy quantum annealer with \autoref{thm:concentration_renyi}.

Indeed, we conclude that:
\begin{prop}\label{prop:concentration_noisy_annealer}
For $0< q\leq \tfrac{1}{2}$ and $T>0$ let
\begin{align}
    r_2=2\frac{1-q}{\log(q^{-1})}\,.
\end{align}
and $h(T)$ be
\begin{align}
h(T)=e^{-r_2T}\log\left(\frac{ 1+2 \left(q-q^2\right)^{\frac{1}{2}}}{4\left(q-q^2\right)}\right)+\frac{(2q-1)(1-e^{-r_2T}rT-e^{-r_2T})}{(q(1-q))^{\tfrac{1}{2}}r_2^2T}.
\end{align}
Furthermore, let $\mathcal{T}_t$ be defined as in \autoref{prop:entropy_decay_annealer} and $f(t)=(1-t)$. Then for the initial state $\ket{+}^{\otimes n}$ we have:
\begin{align}
\mathbb{P}_{\mathcal{T}_T(\ketbra{+}^{\otimes n})}(H_I\leq(\tr(\left[\tau_q^{\otimes n}H_I\right]-2^{-\frac{1}{2}}(h(T)+\epsilon)^{1/2}\|H_I\|_{\operatornamewithlimits{Lip}}n)\mathbb{I})\leq \operatorname{exp}\left(-\frac{\epsilon n}{2}\right).
\end{align}
\end{prop}
\begin{proof}
As proved in Eq.~\eqref{equ:bound_entropy_annealer_beyond2}, it follows that at time $T$ we have that
\begin{align}
n^{-1}D_2(\mathcal{T}_T(\ketbra{+}^{\otimes n})\|\tau_q^{\otimes n})\leq h(T)
\end{align}
and we have the concentration inequality in Eq.~\eqref{equ:concentration_fixed_noise} for $\tau_q^{\otimes n}$. 

We now pick the parameter $r=2^{-\frac{1}{2}}\|H_I\|_{{L}}(h(T)+\epsilon)^{1/2}n$ for the concentration inequality. It then follows from \autoref{thm:concentration_renyi} that:
\begin{align}
    \mathbb{P}_{\mathcal{T}_T(\ketbra{+}^{\otimes n})}(H_I\leq(\tr(\left[\tau_q^{\otimes n}H_I\right]-2^{-\frac{1}{2}}(h(T)+\epsilon)^{1/2}\|H_I\|_{{L}}n)\mathbb{I})\leq \operatorname{exp}\left(\frac{n}{2}(h(T)-(h(T)-
    \epsilon))\right),
\end{align}
which yields the claim.
\end{proof}

\section{Bounds on purity and higher moments}\label{sec:succes_distillation}
We will now obtain upper bounds on $\tr\left[\rho^k\right]$ for $\rho$ the output of a noisy circuit  and $k\geq2$. The motivation for such bounds comes from understanding the success probability of virtual distillation or cooling protocols~\cite{Huggins2021,Koczor2021}. Roughly speaking, these protocols have as their goal to prepare the quantum state $\rho^k/\tr\left[\rho^k\right]$ from $k$ copies of $\rho$. As explained before in \autoref{sec:virtual_cooling}, the success probability of these protocols is $\tr\left[\rho^k\right]$. We will prove that at constant depth the success probability becomes exponentially small in system size. The first observation we make is that $k\mapsto \tr\left[\rho^k\right]$ is monotonically decreasing in $k$. Thus it suffices to show that the purity $\tr\left[\rho^2\right]$ is exponentially small at constant depth under noise.
We start with the following Lemma:
\begin{lem}\label{lem:bound_purity_noise}
Let $\tau_q=q\ketbra{0}+(1-q)\ketbra{1}$ and assume w.l.o.g. that $q\leq\tfrac{1}{2}$. Then for for any state $\rho\in\mathcal{S}_V$ with $n=|V|$ such that 
\begin{align}\label{equ:rho_sigma_purity}
D_2(\rho\|\tau_q^{\otimes n})\leq(1-\epsilon-\log(2(1-q))n\,,
\end{align}
 we have that 
\begin{align}
\tr\left[\rho^2\right]\leq 2^{-\epsilon n}.
\end{align}
\end{lem}
\begin{proof}
   First note that we have:
\begin{align}\label{equ:purityrenyi}
\tr\left[\rho^2\right]=2^{-n+D_2(\rho\|\mathbb{I}/2^n)}.
\end{align}
Thus, the claim follows if we show that Eq.~\eqref{equ:rho_sigma_purity} implies that $D_2(\rho\|\mathbb{I}/2^n)\leq (1-\epsilon)n$. 
From the data-processed triangle inequality \cite[Theorem 3.1]{christandl_relative_2017} we obtain that 
\begin{align}
    &D_2(\rho\|\mathbb{I}/2^n)\leq D_2(\rho\|\tau_q^{\otimes n})+D_{\infty}(\tau_q^{\otimes n}\|\mathbb{I}/2^n)=D_2(\rho\|\tau_q^{\otimes n})+\log(2(1-q))\,n
\end{align}
and so it follows from Eq.~\eqref{equ:rho_sigma_purity} that
\begin{align*}
    D_2(\rho\|\mathbb{I}/2^n)\leq (1-\epsilon)n
\end{align*}
and we obtain the bound by inserting the equation above into Eq.~\eqref{equ:purityrenyi}.
\end{proof}
As before, in the case of unital noise ($q=1/2$) we are able to obtain statements that are independent from the circuit being implemented. 
Moreover, the limitations are even more striking than for our concentration bounds, as summarized in the statement below:
\begin{prop}\label{prop:purity_vanishes}
    Let $\cN_{V}$ be a depth $L$ unitary circuit interspersed by a unital noise channel $\cN$ such that:
    \begin{align*}
        D_2\Big(\cN_{V}(\rho)\Big\|\frac{I}{2^n}\Big)\leq (1-r_2)D\Big(\rho\Big\|\frac{I}{2^n}\Big).
    \end{align*}
Then for any initial state $\rho$ and $k\geq 2$
    \begin{align}\label{equ:bound_purity2}
        \tr\left[\cN_{V}(\rho)^k\right]&\leq\tr\left[\cN_{V}(\rho)^2\right]\nonumber\le\operatorname{exp}(-\log(2)(1-(1-r_2)^{L})n).
    \end{align}

\end{prop}
\begin{proof}
The result immediately follows from combining \autoref{lem:bound_purity_noise} with the fact that the output satisfies:
\begin{align*}
    D_2\left(\cN_{V}(\rho)\|\frac{I}{2^n}\right)\leq (1-r_2)^Ln.
\end{align*}
\end{proof}
Thus, we see from Eq.~\eqref{equ:bound_purity2} that unless $r_2=\cO(n^{-1})$, the purity will already be exponentially small in system size and the distillation or cooling protocols will be unlikely to accept. We conclude that in the case of unital, they are only effective for circuits that have a constant number of expected errors.

In the case of $q<\tfrac{1}{2}$, the success probability will only be exponentially at depths after which the R\'enyi-2 divergence has contracted by more than $\log(2q)$. We can use results like~\autoref{prop:entropy_decay_annealer} and~\autoref{cor:QAOA_rel_ent} to estimate when this happens. However, in these cases we expect that the bounds we currently have only predict that the depth at which the purity becomes exponentially are of order $\cO(p_\alpha^{-1})$. 

Let us illustrate this more concretely with noisy annealers.
For the case of a linear schedule we obtain under the same conditions as for \autoref{prop:concentration_noisy_annealer} that:
\begin{prop}\label{prop:purity_noisy_annealer}
For $0< q\leq \tfrac{1}{2}$ and $T>0$ let
\begin{align}
    r_2=2\frac{1-q}{\log(q^{-1})}\,.
\end{align}
and $h(T)$ be
\begin{align}
h(T)=e^{-r_2T}\log\left(\frac{ 1+2 \left(q-q^2\right)^{\frac{1}{2}}}{4\left(q-q^2\right)}\right)+\frac{(2q-1)(1-e^{-r_2T}rT-e^{-r_2T})}{(q(1-q))^{\tfrac{1}{2}}r_2^2T}.
\end{align}
Furthermore, let $\mathcal{T}_t$ be defined as in \autoref{prop:entropy_decay_annealer} and $f(t)=(1-t)$. For the initial state $\ket{+}^{\otimes n}$ let $T$ be large enough for $h(T)\leq 1-\log(2(1-q))-\epsilon$ to hold for some $\epsilon>0$. Then:
\begin{align}
\tr\left[\mathcal{T}_T(\ketbra{+}^{\otimes n})^k\right]&\leq\tr\left[\mathcal{T}_T(\ketbra{+}^{\otimes n})^2\right]\nonumber\le\operatorname{exp}(-\epsilon n).
\end{align}
\end{prop}
\begin{proof}
This statement immediately follows from \autoref{lem:bound_purity_noise} and the fact that $n^{-1}D_2(\mathcal{T}_T(\ketbra{+}^{\otimes n})\|\tau_q^{\otimes n})\leq h(T)$, as we showed in Eq.~\eqref{equ:bound_entropy_annealer_beyond2}.
\end{proof}
As explained in the main text, such bounds can be applied to bound the probability that virtual distillation protocols work.

\section{Proof of \autoref{thm:poicarconsequ}}\label{app:proofThm3.1}
In this section, we prove the consequences of the $(2,\infty)$-Poincar\'{e} inequality stated in \autoref{thm:poicarconsequ}, which we restate below for clarity of the exposition:

\begin{theorem}
Assume that the state $\sigma\in\mathcal{S}_V$ satisfies a $(2,\infty)$-Poincar\'{e} inequality with constant $C>0$. Then, 
\begin{itemize}
    \item[(i)] \underline{Non-commutative transport-variance inequality}: for any two states $\rho_1,\rho_2\in\mathcal{S}_V$ with corresponding densities $X_j:=\sigma^{-\frac{1}{2}}\rho_j\sigma^{-\frac{1}{2}}$,
\begin{align}
W_1(\rho_1,\rho_2)\le \,\sqrt{C|V|\,} \Big(\|X_1-\mathbb{I}\|_{\sigma}+\|X_2-\mathbb{I}\|_{\sigma}\Big)\,.
\end{align}
\item[(ii)] \underline{Measured transport-variance inequality}: denote by $\mu_\sigma\in \mathcal{P}_V$ the probability measure induced by the measurement of $\sigma$ in the computational basis. Then, for any $\nu<\!<\mu_\sigma$, 
\begin{align*}
W_1(\nu,\mu_\sigma)\le \sqrt{C\,|V|\,\operatorname{Var}_{\mu_\sigma}(d\nu/d\mu_\sigma)}\,.
\end{align*}
Moreover, for any two sets $A,B\subset [d]^V$, their Hamming distance $d_H(A,B)$ satisfies the following symmetric concentration inequality
\begin{align}\label{classisoperimetric1}
    d_H(A,B)\le \sqrt{C\,|V|}\,\big(\mu_\sigma(A)^{-\frac{1}{2}}+\mu_\sigma(B)^{-\frac{1}{2}}\big)\,.
\end{align}
\item[(iii)] \underline{Concentration of observables}: for any observable $O\in\mathcal{O}_V$ and $r>0$,
    \begin{align}
        \mathbb{P}_{\sigma}\big(|O-\langle O\rangle_{\sigma}|\ge r\big)\le  \frac{C|V|\,\|O\|_L^2}{r^2}\,.
    \end{align}
\end{itemize}
\end{theorem}
\begin{proof}
(i) By the Cauchy-Schwarz inequality, we have for $X=\sigma^{-\frac{1}{2}}\rho\sigma^{-\frac{1}{2}}$ that
\begin{align}
    W_1(\rho,\sigma)&=\sup_{\|H\|_L\le 1}\tr\big[(\rho-\sigma)H\big]\nonumber\\
  &  =\sup_{\|H\|_L\le 1}\,\langle X-\mathbb{I},\,H-\mathbb{I}\rangle_\sigma\nonumber\\
  &\le  \sup_{\|H\|_L\le 1}\,\|X-\mathbb{I}\|_{\sigma}\,\|H-\mathbb{I}\|_{\sigma}\nonumber\\
  &\le \sup_{\|H\|_L\le 1}\|X-\mathbb{I}\|_{\sigma}\,\operatorname{Var}_\sigma(H)^{\frac{1}{2}}\label{eq:kmsgns}\\
  &\le \sqrt{C\,|V|}\,\|X-\mathbb{I}\|_{\sigma}\,,\nonumber
\end{align}
where \eqref{eq:kmsgns} follows from \eqref{KMStoGNS}. Therefore by the triangle inequality, for any two states $\rho_1,\rho_2$ with corresponding densities $X_j=\sigma^{-\frac{1}{2}}\rho_j\sigma^{-\frac{1}{2}}$:
\begin{align*}
    W_1(\rho_1,\rho_2)\le \,\sqrt{C|V|\,}\Big(\|X_1-\mathbb{I}\|_{\sigma}+\|X_2-\mathbb{I}\|_{\sigma}\Big)\,.
\end{align*}

(ii) Since the quantum variance $\operatorname{Var}_{\sigma}(O_F)$ and the classical variance $\operatorname{Var}_{\mu_\sigma}(F)$ coincide for a classical Lipschitz function $F$ with $O_F:=\sum_{\eps\in [d]^V}F(\eps)|\eps\rangle\langle\eps|$, we have that
\begin{align}
    \operatorname{Var}_{\mu_\sigma}(F)\le C\,|V|\,\|F\|_L^2\,,
\end{align}
where we further used that classical and quantum Lipschitz constants coincide, i.e.~$\|F\|_L=\|O_F\|_L$, see \cite[Proposition 7]{de2020quantum}. Next, by the Cauchy-Schwarz inequality, we have for all $\nu<\!<\mu_\sigma$ and $g:=\frac{d\nu}{d\mu_\sigma}$ that
\begin{align*}
    W_1(\nu,\mu_\sigma)&=\sup_{\|F\|_L\le 1}\nu(F)-\mu_\sigma(F)\\
  &  =\sup_{\|F\|_L\le 1}\,\mu_\sigma\big((g-1)(F-1)\big)\\
  &\le  \sup_{\|F\|_L\le 1}\,\|g-1\|_{L_2(\mu_\sigma)}\,\|F-1\|_{L_2(\mu_\sigma)}\\
  &= \sup_{\|F\|_L\le 1}\Big(\operatorname{Var}_{\mu_\sigma}(g)\,\operatorname{Var}_{\mu_\sigma}(F)\Big)^{\frac{1}{2}}\\
  &\le \sqrt{C\,|V|\,\operatorname{Var}_{\mu_\sigma}(g)}\,.
\end{align*}
The proof of \eqref{classisoperimetric1} is standard \cite{marton1986simple}: denote by $\nu_A$, resp.~$\nu_B$, the probability measures
\begin{align}
    \nu_A(C):=\frac{\mu_\sigma(A\cap C)}{\mu_\sigma(A)},\,\qquad \nu_B(C):=\frac{\mu_\sigma(B\cap C)}{\mu_\sigma(B)}\,.
\end{align}
Then by the dual formulation of the Wasserstein distance in terms of couplings, we have 
\begin{align*}
    d_H(A,B)&\le W_1(\nu_A,\nu_B)\le W_1(\nu_A,\mu_\sigma)+W_1(\nu_B,\mu_\sigma)\\
    &\le  \sqrt{C\,|V|} \Big(\operatorname{Var}_{\mu_\sigma}(d\nu_A/d\mu_\sigma)^{\frac{1}{2}}+\operatorname{Var}_{\mu_\sigma}(d\nu_B/d\mu_\sigma)^{\frac{1}{2}}\Big)\\
    &=\sqrt{C\,|V|}\,\Big(\mu_\sigma(A)^{-\frac{1}{2}}+\mu_\sigma(B)^{-\frac{1}{2}}\Big)\,,
\end{align*}
where the last line follows from bounding the variance 
\begin{align}
    \operatorname{Var}_{\mu_\sigma}(d\nu_A/d\mu_\sigma)=\int\,\Big( \frac{d\nu_A}{d\mu_{\sigma}}\Big)^2d\mu_\sigma-1= \int \frac{1_A(x)}{\mu_\sigma(A)^2}d\mu_\sigma(x)-1\le \frac{1}{\mu_\sigma(A)}
\end{align}
and similarly for $\operatorname{Var}_{\mu_\sigma}(d\nu_B/d\mu_\sigma)$.

(iii) is a direct consequence of the $(2,\infty)$-Poincar\'{e} inequality and Chebyshev's inequality.

\end{proof}

\section{Controlling the Lipschitz constant}\label{appendixLipschitzcontrol}

In this appendix, we derive the bounds on the Lipschitz constant of observables evolving according to a local continuous and discrete-time evolution, namely \autoref{shortdepth}, \autoref{noiselesslight-cone} and \autoref{conttimepoinca}. 
We start by proving the {\texorpdfstring{$(2,\infty)$}{(2,inf)}}-Poincar\'{e} inequality for product states. In fact, we will prove a slight refinement of it. Let us start by defining for any $O\in\mathcal{O}_V$ and any $v\in V$: 
\begin{equation*}
    \partial_v O := 2\min\left\{\left\|O - \mathbb{I}_v\otimes O_{v^c}\right\|_\infty:O_{v^c}\in\mathcal{O}_{v^c}\right\}\,.
\end{equation*}
By definition, we hence have that $\left\|O\right\|_L = \max_{v\in V}\partial_vO$.

\begin{lem}\label{lem22poincarreproduct}
For any product state $\rho\in\mathcal{S}_V$ and all $O\in\mathcal{O}_V$,
\begin{align*}
    \operatorname{Var}_\rho(O)\le \sum_{v\in V}\left(\partial_v O\right)^2 \le \left|V\right|\left\|O\right\|_L^2\,.
\end{align*}
\end{lem}

\begin{proof}
We fix an arbitrary ordering $\{1,...,n\}$ of the vertices $V$.
For any $i\in[n]$, let $O_{i^c}\in\mathcal{O}_{i^c}$ satisfy
\begin{equation*}
    \partial_i O = 2\left\|O - \mathbb{I}_i\otimes O_{i^c}\right\|_\infty\,.
\end{equation*}
Given a subregion $A\subseteq V$, we denote $\langle O\rangle_{\rho_A}:=\tr_{A}[\rho_A O]\otimes\mathbb{I}_A$. Then, by a telescopic sum argument:
\begin{align*}
    \operatorname{Var}_\rho(O)&=\tr\Big[\rho\Big(\sum_{i=1}^n\langle O\rangle_{\rho_{1...i-1}}-\langle O\rangle_{\rho_{1...i}}\Big)^2\Big]\\
    &=\sum_{i,j=1}^n\,\tr\big[\rho\,(\langle O\rangle_{\rho_{1...i-1}}-\langle O\rangle_{\rho_{1...i}})\,(\langle O\rangle_{\rho_{1...j-1}}-\langle O\rangle_{\rho_{1...j}})\big]\\
    &\overset{(1)}{=}\sum_{i=1}^n\,\tr\big[\rho\,(\langle O\rangle_{\rho_{1...i-1}}-\langle O\rangle_{\rho_{1...i}})^2\big]\\
    &=\sum_{i=1}^n\,\tr\big[\rho\,(\langle O - \mathbb{I}_i\otimes O_{i^c}\rangle_{\rho_{1...i-1}}-\langle O - \mathbb{I}_i\otimes O_{i^c}\rangle_{\rho_{1...i}})^2\big]\\
    &\le\sum_{i=1}^n\left\|\langle O - \mathbb{I}_i\otimes O_{i^c}\rangle_{\rho_{1...i-1}}-\langle O - \mathbb{I}_i\otimes O_{i^c}\rangle_{\rho_{1...i}}\right\|_\infty^2\\
    &\le4\sum_{i=1}^n\left\|O - \mathbb{I}_i\otimes O_{i^c}\right\|_\infty^2 = \sum_{i=1}^n\left(\partial_iO\right)^2\,.
\end{align*}
In $(1)$ above, we used the orthogonality relation that for any $i\ne j$
\begin{align}
   \tr\big[\rho\,(\langle O\rangle_{\rho_{1...i-1}}-\langle O\rangle_{\rho_{1...i}})(\langle O\rangle_{\rho_{1...j-1}}-\langle O\rangle_{\rho_{1...j}})\big] =0\,
\end{align}
since $\rho\equiv \bigotimes_{v\in V}\rho_v$ is assumed to be a tensor product.
\end{proof}

We will need the following technical result in what follows:

\begin{prop}\label{prop22poincarcircuit}
Let $v\in V$, and let $I_v$ be the future light-cone of $v$ with respect to the quantum channel $\mathcal{N}_V$ on $\mathcal{S}_V$.
Then, for any $O\in\mathcal{O}_v$,
\begin{equation*}
    \partial_v\mathcal{N}_V^\dag(O) \le 2\sum_{w\in I_v}\partial_w O \le 2\left|I_v\right|\left\|O\right\|_L\,.
\end{equation*}
\end{prop}

\begin{proof}
For any $w\in V$, let $O_{w^c}\in\mathcal{O}_{w^c}$ such that
\begin{equation*}
    \partial_wO = 2\left\|O - \mathbb{I}_w\otimes O_{w^c}\right\|_\infty\,.
\end{equation*}
Let $\left|I_v\right|=k$, and let us label with $\{1,\,\ldots,\,k\}$ the elements of $I_v$.
The observable $\mathcal{N}^\dag\left(\mathbb{I}_{I_v}\otimes\frac{\mathrm{Tr}_{I_v}O}{d^{|I_v|}}\right)$ does not act on $v$, therefore
\begin{align}
    \partial_v\mathcal{N}^\dag(O) &\le 2\left\|\mathcal{N}^\dag\left(O - \mathbb{I}_{I_v}\otimes\frac{\mathrm{Tr}_{I_v}O}{d^{|I_v|}}\right)\right\|_\infty\nonumber\\
    &\le 2\left\|O - \mathbb{I}_{I_v}\otimes\frac{\mathrm{Tr}_{I_v}O}{d^{|I_v|}}\right\|_\infty\nonumber\\
    &\le 2\sum_{i=1}^k\left\|\mathbb{I}_{1\ldots i-1}\otimes\frac{1}{d^{i-1}}\,\mathrm{Tr}_{1\ldots i-1}O - \mathbb{I}_{1\ldots i}\otimes\frac{1}{d^{i}}\,\mathrm{Tr}_{1\ldots i}O\right\|_\infty\nonumber\\
    &= 2\sum_{i=1}^k\left\|\mathbb{I}_{1\ldots i-1}\otimes\frac{1}{d^{i-1}}\,\mathrm{Tr}_{1\ldots i-1}\left[O - \mathbb{I}_i\otimes O_{i^c}\right]  - \mathbb{I}_{1\ldots i}\otimes\frac{1}{d^{i}}\,\mathrm{Tr}_{1\ldots i}\left[O - \mathbb{I}_i\otimes O_{i^c}\right]\right\|_\infty\nonumber\\
    &\le 4\sum_{i=1}^k\left\|O - \mathbb{I}_i\otimes O_{i^c}\right\|_\infty = 2\sum_{i=1}^k\partial_iO\,.\nonumber 
\end{align}
\end{proof}

Next, we consider the noisy circuit introduced in \eqref{equ:localchannel}. For any noisy gate $\cN_{\ell,e}$, we denote by $\sigma_{\ell,e'}$ the environment state of the copy $e'$ of the set $e$, and by $\mathcal{U}_{\ell,\{e,e'\}}$ the unitary dilation of $\cN_{\ell,e}$ acting on set $e$ and its copy $e'$, so that
\begin{align*}
    \mathcal{N}_{\ell,e}(\rho)=\tr_{e'}\,\big(\mathcal{U}_{\ell,\{e,e'\}}(\rho\otimes\sigma_{\ell,e'})\big)\,.
\end{align*}
We also denote by $\mathcal{U}_{VA}$ the composition of the tensor products of dilations $\mathcal{U}_{\ell,\{e,e'\}}$, where the system $A$ represents the total environment resulting from all the dilations previously defined. 
In other words, defining $\sigma_A:=\bigotimes_{\ell,e}\sigma_{\ell,e'}$, we have
\begin{align*}
    \cN_V(\rho)=\tr_A\big[\mathcal{U}_{VA}(\rho\otimes \sigma_A)\big]\equiv \rho_{\operatorname{out}}\,.
\end{align*}
We denote by $I_{\mathcal{U}_{VA}}$ the light-cone of $\mathcal{U}_{VA}$ with respect to the decomposition 
\begin{align}\label{decompositionVA}
\cH_V\otimes \cH_A\equiv \bigotimes_{v\in V}\cH_v\otimes \bigotimes_{\ell,e}\cH_{\ell,e'}\,.
\end{align}

\begin{prop}\label{prop:Poinc22}
For any $O\in\mathcal{O}_V$ and any product state $\rho\in\mathcal{S}_V$ with $\rho_{\operatorname{out}}:=\cN_V(\rho)$:
\begin{align*}
    \operatorname{Var}_{\rho_{\operatorname{out}}}(O)&\le 4\,\|O\|_L^2\,\Big(|V|\,I_{\cN_V}^2+\max_\ell|E_\ell|\,\sum_{\ell=1}^L\,\max_{e\in E_\ell}\,I(e,L-\ell)^2\Big)\,,
\end{align*}
 where given a set $e\in E_\ell$ and $m\in\mathbb{N}$, $I(e,L-\ell)$ denotes the set of all vertices in $V$ in the light-come of the set $e$ for the circuit constituted of the last $L-\ell$ layers of $\mathcal{N}_V$.
\end{prop}

\begin{rem}
In the noiseless setting where there are no ancilla systems, by a closer look into the proof below, we can get rid of the sum over layers and hence recover the bound in \autoref{noiselesslight-cone}. 
\end{rem}

\begin{proof}

Given the tensor product input state $\rho$ and for any $O\in\mathcal{O}_V$, we consider the variance 
\begin{align*}
    \operatorname{Var}_{\rho_{\operatorname{out}}}(O)& =\tr\big[\cN_V(\rho)\,(O-\tr[\cN_V(\rho) O]\,\mathbb{I})^2\big]\\
    & =\tr\big[(\rho\otimes \sigma_A) \,\mathcal{U}_{VA}^\dagger(O-\tr[(\rho\otimes\sigma_A)\, \mathcal{U}_{VA}^\dagger(O)]\,\mathbb{I})^2\big]\\
    & =\operatorname{Var}_{\rho\otimes\sigma_A}\big[\mathcal{U}_{VA}^\dagger(O)\big]\,.
    \end{align*}
    Next, we consider the Lipschitz constant $\|.\|^{(VA)}_{L}$ corresponding to the decomposition \eqref{decompositionVA}. In particular, the system $\cH_V\otimes\cH_A$ is constituted of at most $|V|\big(1+\frac{L}{2}\big)$ particles. Since the state $\rho\otimes\sigma_A$ is a tensor product state with respect to the above decomposition, we have the following $(2,\infty)$-Poincar\'{e} inequality from \autoref{lem22poincarreproduct}:
\begin{align*}
    \operatorname{Var}_{\rho_{\operatorname{out}}}(O)&\le \sum_{\omega\in VA}\,\big(\partial_\omega \mathcal{U}^\dagger_{VA}(O)\big)^2=\sum_{v\in V}\,\big(\partial_v \mathcal{U}^\dagger_{VA}(O)\big)^2\,+\,\sum_{a\in A}\big(\partial_a \mathcal{U}^\dagger_{VA}(O)\big)^2
\end{align*}
For the first sum, we get using \autoref{prop22poincarcircuit}:
\begin{align*}
    \sum_{v\in V}\,\big(\partial_v \mathcal{U}^\dagger_{VA}(O)\big)^2&\le 4\,\sum_{v\in V}\,\Big(\sum_{\omega\in I_v^{\mathcal{U}_{VA}}}\partial_\omega O\Big)^2\\
&=    4\,\sum_{v\in V}\,\Big(\sum_{\omega\in I_v^{\mathcal{U}_{VA}}\backslash A}\partial_\omega O\Big)^2
    \\
    &\le 4|V|\,\max_{v\in V} |I_v^{\mathcal{U}_{VA}}\backslash A|^2\,\|O\|_L^2\\
    &=4|V|\,\max_{v\in V} |I_v^{\mathcal{N}_V}|^2\,\|O\|_L^2\\
    &=4|V|\,I^2_{\mathcal{N}_V}\,\|O\|_L^2\,.
\end{align*}
The second sum on the other hand can be controlled as follows: first for any layer ancilla $a$, denote by $\ell_a$ the layer at which $a$ is brought, and decompose the dilation $\mathcal{U}_{VA}$ as
\begin{align*}
    \mathcal{U}_{VA}=\mathcal{U}_{VA}^{[\ell_a,L]}\circ \mathcal{U}_{VA}^{[1,\ell_a-1]}\,,
\end{align*}
where the first subcircuit $\mathcal{U}_{VA}^{[1,\ell_a-1]}$ corresponds to the first $\ell_a-1$ layers, and the second subcircuit $\mathcal{U}_{VA}^{[\ell_a,L]}$ corresponds to the other layers. Then,
\begin{align*}
    \partial_a \mathcal{U}_{VA}^{\dagger}(O)=2\min_{O_{a^c}}\|(\mathcal{U}_{VA}^{[1,\ell_a-1]})^{\dagger}\circ (\mathcal{U}_{VA}^{[\ell_a,L]})^{\dagger}(O)-O_{a^c}\otimes \mathbb{I}_a\|_\infty\le \partial_a\,(\mathcal{U}_{VA}^{[\ell_a,L]})^\dagger(O)
\end{align*}
where the inequality arises by choosing $O_{a^c}=(\mathcal{U}_{VA}^{[1,\ell_a-1]})^\dagger(\widetilde{O}_{a^c})$ with $\widetilde{O}_{a^c}$ the optimizer of $\partial_a(\mathcal{U}_{VA}^{[\ell_a,L]})^\dagger(O)$. Therefore,
\begin{align*}
    \sum_{a\in A}\big(\partial_a \mathcal{U}^\dagger_{VA}(O)\big)^2&\le \sum_{a\in A}\,\big((\mathcal{U}_{VA}^{[\ell_a,L]})^\dagger(O)\big)^2\\
    &\le 4\, \sum_{a\in A}\Big(\sum_{\omega\in I_a^{\mathcal{U}_{VA}^{[\ell_a,L]}}\backslash A}\,\partial_\omega O\Big)^2\\
    &\le 4\|O\|_L^2\,\sum_{a\in A}\,|I_a^{\mathcal{U}_{VA}^{[\ell_a,L]}}\backslash A|^2\\
    &\le 4\|O\|_L^2\,\max_\ell|E_\ell|\,\sum_{\ell=1}^L\,\max_{e\in E_\ell}\,I(e,L-\ell)^2\,,
\end{align*}
where the second inequality follows again from \autoref{prop22poincarcircuit}.
\end{proof}

\section{Controlling the Lipschitz constant for Hamiltonian dynamics}\label{sec-liebrobin}

In this appendix, we prove \autoref{eqlipliebrob} (see \cite[Proposition B.2.]{rouze2021learning} for a similar statement). We first recall the following equivalent formulation of the Wasserstein distance due to \cite{de2020quantum}:

\begin{align}\label{wassersteinnormDP}
{W}_{1}(\rho,\sigma)=\frac{1}{2}\,\min \left\{\sum_{i=1}^{n} \|X^{(i)}\|_1:\rho-\sigma=\sum_{i=1}^{n} X^{(i)}, X^{(i)} \in \mathcal{O}_V^T, \operatorname{tr}_{i} [X^{(i)}]=0\right\}\,,
\end{align}
 where we recall that $\mathcal{O}_V^T$ denotes the set of self-adjoint, traceless observables.

\begin{prop}
Assume that the continuous-time evolution $\{\mathcal{U}_V(t)\}_{t\ge 0}$ defined on the graph $G=(V,E)$ with $|V|=n$ satisfies the bound in \autoref{thm:LR}. Then, for any $H\in\mathcal{O}_V$,
\begin{align}\label{eqrobinson}
    \|\mathcal{U}_V(t)^\dagger(H)\|_L
     \le \left(2(i_0-1)+\frac{4M}{2D-1}\sum_{i=i_0}^{n} d(i)^{\delta-1}\,e^{vt-d(i)}\right)\|H\|_L\,,
\end{align}
where $\operatorname{dist}(\{1\},\{i,\cdots,n\})\equiv d(i)$, and $i_0$ stands for the first vertex such that $d(i_0)\ge 2\delta-1$.

\end{prop}

\begin{proof}
From \cite{de2020quantum}, the Wasserstein distance $W_{1}$ arises from a norm $\|.\|_{W_1}$, i.e. $W_{1}(\rho,\sigma)=\|\rho-\sigma\|_{W_1}$. Moreover, the norm $\|.\|_{W_1}$ is uniquely determined by its unit ball, which in turn is the convex hull of the set of the differences between couples of neighboring quantum states: 
\begin{align*}
    \mathcal{N}_n=\bigcup_{i\in V}\,\mathcal{N}_n^{(i)}\,,~~~~~\mathcal{N}_n^{(i)}=\{\rho-\sigma:\,\rho,\sigma\in\mathcal{S}_{V},\,\operatorname{tr}_i(\rho)=\operatorname{tr}_i(\sigma)\}\,.
    \end{align*}
    Now by convexity, the contraction coefficient for this norm is equal to
    \begin{align*}
        \|\mathcal{U}_V(t)\|_{W_1\to W_1}=\max\big\{\|\mathcal{U}_V(t)(X)\|_{W_1}:\,X\in \mathcal{O}_V^T,\,\|X\|_{W_1}\le 1  \big\}=\max_{X\in\mathcal{N}_n}\|\mathcal{U}_V(t)(X)\|_{W_1}\,.
    \end{align*}
   Let then $X\in\mathcal{N}_n$. By the expression \eqref{wassersteinnormDP}, and choosing without loss of generality an ordering of the vertices such that $\operatorname{tr}_1(X)=0$, we have 
    \begin{align}\label{equ:LR_estimates}
        \|\mathcal{U}_V(t)(X)\|_{W_1}&\le \frac{1}{2}\sum_{i=1 }^{n}\Big\|\frac{I}{d^{i-1}}\otimes \operatorname{tr}_{1\cdots i-1}\circ \,\mathcal{U}_V(t)(X)-\frac{I}{d^i}\otimes \operatorname{tr}_{1\cdots i}\circ\, \mathcal{U}_V(t)(X)\Big\|_1\nonumber\\
        &=\frac{1}{2}\,\sum_{i=1}^n\,\Big\|\,\int\,d\mu(U_i)\,\operatorname{tr}_{1\cdots i-1}\circ (\mathcal{U}_V(t)(X)-U_i\,\mathcal{U}_V(t)(X)U_i^\dagger)\Big\|_1\nonumber\\
        &\le \frac{1}{2}\sum_{i=1}^n\int\,d\mu(U_i)\,\|[U_i,\operatorname{tr}_{1\cdots i-1}\circ\,\mathcal{U}_V(t)(X)]\|_1\nonumber\\
        &\le \sum_{i=1}^n\|\operatorname{tr}_{1\cdots i-1}\circ\,\mathcal{U}_V(t)(X)\|_1\nonumber\\
        &\overset{(1)}{=} \sum_{i=1}^n\|\operatorname{tr}_{1\cdots i-1}\circ(\mathcal{U}_V(t)-\mathcal{U}_{\{i-k,\cdots, n\}}(t))(X)\|_1
        \end{align}
where $\mu$ denotes the Haar measure on one qudit, and where $(1)$ follows from the fact that $\operatorname{tr}_1(X)=0$, with $ \mathcal{U}_{\{i-k,\cdots, n\}}(t)$ defined as in \autoref{thm:LR} with $k<i-1$. Next, by the variational formulation of the trace distance and \autoref{thm:LR}, we have for $i\geq i_0$ that
\begin{align*}
    \|\operatorname{tr}_{1\cdots i-1}\circ(\mathcal{U}_V(t)-\mathcal{U}_{\{i-k,\cdots, n\}}(t))(X)\|_1&=\max_{\|O_{i\cdots n}\|_\infty\le 1}\,\big|\operatorname{tr}\big[X(\mathcal{U}_V(t)^\dagger-\mathcal{U}_{\{i-k,\cdots, n\}}(t)^\dagger)(O_{i\cdots n}) \big]\big|\\
    &\le \,\max_{\|O_{i\cdots n}\|_\infty\le 1}\,\|(\mathcal{U}_V(t)^\dagger-\mathcal{U}_{\{i-k,\cdots, n\}}(t)^\dagger)(O_{i\cdots n}) \|_\infty\|X\|_1\\
    &\le  \frac{2M}{2D -1}\,d_{i,k}^{\delta-1}\,e^{v t-d_{i,k}}\,\|X\|_1\\
    &\overset{(2)}{\le}  \frac{4M}{2D -1}\,d_{i,k}^{\delta-1}\,e^{v t-d_{i,k}}\,\|X\|_{W_1}\,,
    \end{align*}
    where $d_{i,k}:=\operatorname{dist}(\{i\cdots n\},\{1\cdots i-k-1\})$ and $(2)$ follows from \cite[Proposition 6]{de2020quantum}. By picking $k=i-2$ and inserting this estimate into Eq.~\eqref{equ:LR_estimates} for $i\geq i_0$ and the trivial  estimate $\|\operatorname{tr}_{1\cdots i-1}\circ(\mathcal{U}_V(t)-\mathcal{U}_{\{i-k,\cdots, n\}}(t))(X)\|_1\leq2\|X\|_1$ for $i\le i_0-1$, we obtain Eq. \eqref{eqrobinson} by duality.
\end{proof}

\section{Properties of high-energy states of the Max-Cut Hamiltonian}\label{app:MC}
\begin{prop}\label{prop:symmetry}
Under the same hypotheses of \autoref{thm:noiselessMax-Cut}, let $x_{\mathrm{opt}}$ achieve the maximum cut of $G$, and let $X$ be the random outcome obtained measuring $\rho$ in the computational basis.
Then,
\begin{equation}
    \mathbb{P}\left(d_H(X,x_{\mathrm{opt}}) \le \frac{n}{3}\right) = \mathbb{P}\left(d_H(X,\bar{x}_{\mathrm{opt}}) \le \frac{n}{3}\right) \ge \frac{1}{4}\,,
\end{equation}
where $\bar{x}$ denotes the bitwise negation of $x$.
\end{prop}
\begin{proof}
The proof follows the same lines of the proof of \cite[Theorem 1]{bravyi_obstacles_2020}.
Since $G$ is regular and bipartite, we have $2|E| = D\,n$ and $C_{\max} = |E|$, and \eqref{eq:appr} becomes
\begin{equation}
    \mathrm{Tr}\left[\rho\,H\right] \ge |E| - \frac{h\,n}{6}\,.
\end{equation}
We have
\begin{equation}
    \mathrm{Tr}\left[\rho\,H\right] = \mathbb{E}\,C(X)\,,
\end{equation}
and Markov's inequality implies
\begin{equation}\label{eq:PX}
    \mathbb{P}\left(C(X) \ge |E| - \frac{h\,n}{3}\right) \ge \frac{1}{2}\,.
\end{equation}
Since $G$ is bipartite, we have for any $x\in\{0,1\}^n$
\begin{equation}\label{eq:XOR}
    C(x) + C(x_{\mathrm{opt}}\oplus x) = |E|\,,
\end{equation}
where $\oplus$ denotes the sum modulo $2$.
Eqs. \eqref{eq:PX} and \eqref{eq:XOR} imply
\begin{equation}\label{eq:PC}
    \mathbb{P}\left(C(x_{\mathrm{opt}}\oplus X) \le \frac{h\,n}{3}\right)\ge\frac{1}{2}\,.
\end{equation}
The hypothesis \eqref{eq:hypC} and \eqref{eq:PC} imply
\begin{equation}
    \mathbb{P}\left(d_H(X,x_{\mathrm{opt}}) \le \frac{n}{3}\right) + \mathbb{P}\left(d_H(X,\bar{x}_{\mathrm{opt}}) \le \frac{n}{3}\right) \ge \frac{1}{2}\,,
\end{equation}
where we have used that for any $x,\,y\in\{0,1\}^n$ we have
\begin{equation}
    |x\oplus y| = d_H(x,y)\,,\qquad d_H(x,y) + d_H(\bar{x},y) = n\,.
\end{equation}
Since $\rho$ commutes with $\sigma_x^{\otimes n}$, the probability distribution of $X$ is invariant with respect to the negation of all the bits, therefore
\begin{equation}
    \mathbb{P}\left(d_H(X,x_{\mathrm{opt}}) \le \frac{n}{3}\right) = \mathbb{P}\left(d_H(X,\bar{x}_{\mathrm{opt}}) \le \frac{n}{3}\right)\,.
\end{equation}
The claim follows.
\end{proof}

\section{Concentration for error mitigation}\label{sec:concentracion_error}
In this section we will prove~\autoref{thm:theoremmain}, which we restate for the reader's convenience:
\begin{theorem}\label{thm:theoremmain2}
    For an error mitigation observable $X$ 
    \begin{align}
        X:=\sum_{s\in\mathcal{S}}f(s)\,\tr_{A}\left({I_{S}\otimes \ketbra{0}^{\otimes k}M_s}\right)
    \end{align}
    assume that for a given state $\sigma$
    \begin{align}\label{equ:collective_concentration2}
        \mathbb{P}_{\sigma^{\otimes m}}(\|s-\mathbb{E}_{\sigma^{\otimes m}}(s)\|_{\ell_\infty}\geq rn)\nonumber\leq K(m)\textrm{exp}\left(-\frac{cr^2n}{\ell_0^2}\right)
    \end{align}
    holds for some function $K(m)$.
    Furthermore, assume that for $r,\epsilon>0$ given we have for all $1\leq i\leq m$ that $D_2(\mathcal{E}_i(\rho)\|\sigma)= \tfrac{c(r-\epsilon)n}{ml_0}$ and that $f$ is $L_f$-Lipschitz w.r.t. the $\ell_\infty$ norm. Then for $\rho_{\operatorname{out}}=\bigotimes_{i=1}^m\mathcal{E}_i(\rho)$ we have that:
    \begin{align}
        \mathbb{P}_{\rho_{\operatorname{out}}}(|X-f(\mathbb{E}_{\sigma^{\otimes m}}(s))\mathbb{I}|> rL_fn)\le \mathrm{exp}\left(-\frac{c\epsilon n}{\ell_0^2}\right).
    \end{align}
\end{theorem}
\begin{proof}
It follows from the additivity of the R\'enyi divergence and our assumption on $D_2(\mathcal{E}_i(\rho)\|\sigma)$ that
\begin{align}
D_2(\bigotimes_{i=1}^m\mathcal{E}_i(\rho)\|\sigma^{\otimes m})=\sum_iD_2(\mathcal{E}_i(\rho)\|\sigma)\leq \frac{c(r^2-\epsilon)}{l_0^2}.
\end{align}
From~\autoref{lem:transfer_concent} and the concentration inequality in Eq.~\eqref{equ:collective_concentration2} we obtain
\begin{align}\label{equ:collective_concentration_transferred}
    \mathbb{P}_{\rho_{out}}(\|s-\mathbb{E}_{\sigma^{\otimes m}}(s)\|_{\ell_\infty}\geq rn)\nonumber\leq K(m)\textrm{exp}\left(-\frac{c\epsilon}{\ell_0^2}\right)
\end{align}
By our assumption on the function $f$ being Lipschitz, we have that 
\begin{align*}
    \mathbb{P}_{\rho_{\operatorname{out}}}(|X-f(\mathbb{E}_{\sigma^{\otimes m}}(s))\mathbb{I}|\leq rL_fn)\geq \mathbb{P}_{\rho_{out}}(\|s-\mathbb{E}_{\sigma^{\otimes m}}(s)\|_{\ell_\infty}\leq rn)
\end{align*}
and the claim follows.
\end{proof}

\end{document}